\documentclass[a4paper,11pt]{article}
\usepackage[utf8]{inputenc}
\usepackage[noblocks]{authblk}
\usepackage{amsmath, amssymb, amsthm}
\usepackage{lmodern} 
\usepackage{a4wide} 
\usepackage{xcolor}
\usepackage{graphicx}
\usepackage{hyperref}
\usepackage{enumitem}
\usepackage{mathtools}
\usepackage{esint}
\usepackage[normalem]{ulem}
\usepackage{verbatim}
\usepackage{hyperref}

\newcommand{\dd}{\mathrm d}  
\newcommand{\TV}{\mathrm{TV}}
\newcommand{\ROF}{\mathrm{ROF}}

\newtheorem{thm}{Theorem}[section]

\theoremstyle{definition}

\theoremstyle{remark}
\newtheorem{remark}[thm]{Remark}

\usepackage{algorithmic}
\usepackage{algorithm}
\usepackage{subfloat}
\usepackage{subfig}
\usepackage{multirow}
\usepackage{booktabs}
\usepackage{tikz}
\usepackage{pgfplots}
\usepackage{float}

\newtheorem{proposition}{Proposition}

\mathtoolsset{showonlyrefs}

\author{Wojciech G{\' o}rny\thanks{Wojciech G{\' o}rny is with the Faculty of Mathematics, Universit\"at Wien, Oskar-Morgerstern-Platz 1, 1090 Vienna, Austria; and Faculty of Mathematics, Informatics and Mechanics, University of Warsaw, Banacha 2, 02-097 Warsaw, Poland (e-mail:   wojciech.gorny@univie.ac.at).}\quad 
Micha{\l} {\L}asica\thanks{Micha{\l} {\L}asica is with the Institute of Mathematics of the Polish Academy of Sciences, {\'S}niadeckich 8, 00-656 Warsaw, Poland (e-mail: mlasica@impan.pl).}\quad
Alexandros Matsoukas\thanks{Alexandros Matsoukas is with the Department of Mathematics, School of Applied Mathematical and Physical Sciences, National Technical 
University of Athens,  Zografou Campus, 157 80 Athens, Greece (e-mail:  alexmatsoukas@mail.ntua.gr).}}

\title{Adaptive double-phase Rudin--Osher--Fatemi denoising model}

\makeatletter
\newcommand{\namelabel}[1]{%
  \phantomsection
  \renewcommand{\@currentlabel}{#1}
  \label{#1}
}
\makeatother

\makeatletter
\newcommand{\labeltext}[2]{%
  \@bsphack
  \csname phantomsection\endcsname 
  \def\@currentlabel{#1}{\label{#2}}%
  \@esphack
}
\makeatother

\begin{document}

\maketitle

\begin{abstract}
Even though more than 30 years have passed since the seminal Rudin--Osher--Fatemi (ROF) paper on total variation (TV) denoising, it remains relevant due to its simplicity, robustness and interpretability. 
However, it is known to suffer from artifacts such as the staircasing effect. Many variants of the model have been proposed with the aim of countering this. Recently, against the backdrop of immense research output on double-phase problems in the mathematical analysis community, a double-phase type integral functional, comprising of TV and a weighted term of quadratic growth, was suggested as a regularizer for image restoration. 

Here, we propose an adaptive variant of the ROF denoising model based on that regularizer. 
Variable growth of the double-phase functional allows for qualitatively different behavior at image contours, which are captured by an initial ROF reconstruction step. The model is designed to reduce staircasing with respect to the classical ROF model, while preserving the edges of the image in a similar fashion.  We derive a closed-form resolvent formula and adapt the primal-dual Chambolle–Pock scheme for the numerical solution of the model. We also propose a practical noise-dependent parameter prescription and evaluate its performance on synthetic and natural images over a range of noise levels. Compared to established models  with similar interpretability, we observe an improved or similar performance in terms of similarity metrics SSIM, PSNR, and LPIPS, while the staircasing effect is visibly reduced. 
\end{abstract}

\paragraph{keywords}
total variation, image denoising, ROF model, variable growth, double phase

\section{Introduction}
As argued in \cite{LebrunColom}, all photographic images contain noise. Generally speaking, an image $g$ can be represented as a sum 
\[g = u_0 + n,\] 
where $u_0$ is an ideal underlying image, and $n$ is the noise. Thus, noise removal, or \emph{denoising}, is a basic task in image processing. In \cite{ROF}, a denoising model amounting to finding the minimizer of the problem 
\begin{equation} 
\label{rof0}  \min_{u \colon \int|u-g|^2 = d^2} \int |\nabla u|\, \dd x
\end{equation} 
was proposed. The parameter $d>0$ in \eqref{rof0} needs to be suitably chosen based on the image $g$. As observed in \cite{ChambolleLions}, \eqref{rof0} is equivalent to the unconstrained problem (with parameter $\lambda > 0$)
\begin{equation}\label{eq:rof}  \min_{u}\  \int |\nabla u|\, \dd x + \frac{1}{2\lambda}\int|u-g|^2 \,\dd x. 
\end{equation} 
Now often referred to as the ROF model, it is a cornerstone of the field of \emph{variational denoising}. Its fundamental feature is the capability of removing noise while preserving sharp contours (edges) of objects \cite{Casellesetal2007, ChambolleLasica}. Regardless of the development of more sophisticated denoising methods, the ROF model and many of its variants remain relevant 
due to their transparency, interpretability, and computational amenability, with notable recent presence e.g.\ in geoscience \cite{YameiZhou, Niveditha, WangZhang}. In fact, ROF-type models can outperform many other classical models \cite{Roscani20}. In addition, TV-based regularization has been successfully combined with neural network based methods in a variety of settings \cite{TotalDeepVariation, WangQiao23, WangSun24, Mirzoyan25, ZouShoushtari24pre}. However, it also has drawbacks, notably \emph{staircasing}, i.e., emergence of spurious small-scale piecewise constant structure from the noise. Several versions of the ROF model have been proposed in order to alleviate this effect. In those variants, typically the total variation regularizer $TV(v) = \int |\nabla u|\, \dd x$ is replaced with a different one of form $\int \varphi(x, \nabla u) \, \dd x$, aiming to maintain the edge-preserving feature of the ROF model, while adjusting the behavior in the bulk. Mathematically, this corresponds to keeping linear growth of $\varphi(x, \cdot)$ in the regions where the edges are expected, while customizing it elsewhere. Perhaps the simplest variant of this kind (still without direct dependence of $\varphi$ on $x$) is the so-called Huber-TV regularizer \cite{ChambolleLions} with $\varphi(x, \nabla u) = |\nabla u|_\alpha$, $\alpha >0$, where
\begin{equation} \label{huber_func} 
|\xi|_\alpha =
\begin{cases} 
\frac{1}{2\alpha} |\xi|^2 & \text{if } |\xi| \leq \alpha; \\
|\xi| - \frac{\alpha}{2} & \text{if } |\xi| > \alpha.
\end{cases}
\end{equation}
The idea was further elaborated upon by considering variable exponent regularizers $\varphi(x, \nabla u) \sim |\nabla u|^{p(x)}$ in \cite{CLR, BCESV, LiLiPi}. The dependence on $x$ allows customizing the behavior of the regularizer depending on local features of the image, thus allowing for improvement of performance of such regularizers over Huber-TV. Another approach is to use a higher-order regularizer such as the total generalized variation \cite{BKP2010}, \cite{BDRZ2025}.

From a mathematician's point of view, the quantity $|\nabla u|^{p(x)}$ is an example of a \emph{variable growth} integrand. Another typical example in this class is the \emph{double phase} integrand of form $|\nabla u|^p + w(x) |\nabla u|^q$. Mathematical analysis of variational problems involving such integrands has been a hot topic in recent years. In \cite{HarjulehtoHasto2021}, a double phase integrand with $p=1$, $q=2$: 
\begin{equation} \label{dpint}
\varphi(x,\nabla u) = |\nabla u| + w(x) |\nabla u|^2
\end{equation} 
was proposed as a potential regularizer for image restoration. The qualitative behavior of \eqref{dpint} is significantly different from that of the variable exponent integrand. It also has the advantage that the resolvent is given by a simple, explicit formula. Moreover, it behaves nicely under rescaling of the image range, as opposed to $|\nabla u|^{p(x)}$ which becomes weighted after rescaling (hence, its behavior depends on the particular choice of the image range). To the best of our knowledge, this is the first systematic study of an adaptive variational denoising model based on a double-phase regularizer of this form, in which the regularization term is 
truly of variable-growth type. This is in contrast with \cite{CLR}, whose growth is always linear for large arguments and models such as \cite{LiLiPi}, which are always superlinear. The main contribution of the present work is the development of this idea together with an experimental evaluation against established variational models and non-variational non-local denoising methods. For the numerical treatment of the model, we derive an explicit formula for the corresponding resolvent and adapt the primal-dual Chambolle--Pock algorithm. We also propose a practical noise-dependent parameter prescription and investigate the resulting method over multiple noise levels and image datasets, for both grayscale and multichannel data. 


The weight $w(x)$ in \eqref{dpint} plays a similar role to the function $p(x)$ in the variable exponent integrand, deciding whether $\varphi(x, \cdot)$ behaves as a function of linear or superlinear growth at the given point $x$. It needs to be chosen based on the datum $g$ in order to reflect the expected location of edges in the original image $u_0$. In \cite{CLR, BCESV, LiLiPi}, $p(x)$ was chosen to be of the form $P(|\nabla \rho_r * g|)$, where $\rho_r$ is a mollifying kernel and $P$ is a suitable (non-increasing) function with range between~$1$ and~$2$. To improve on this idea, we propose here to choose 
\begin{equation} \label{weight}
 w(x) = W(|\nabla \rho_r * u_{\ROF}|),
\end{equation} 
where $u_{\ROF}$ is the minimizer of the standard ROF model \eqref{eq:rof}, and $W$ may now take values in $[0, \infty[$. Our heuristic justification for this choice is that the ROF model is known to be much more precise at finding the edges of the underlying image than simple mollification of $g$. The mollifcation kernel in \eqref{weight} is to be taken with a very small radius (such as $2$ pixels)---its role is now not to suppress noise, but rather to create a narrow strip around the edges where the regularizer will have linear growth. Then, we chose $W$ as a suitable cutoff-type function (taking large values near $0$ and vanishing for larger arguments) in order to roughly distinguish between "real" edges and the staircase edges created from the noise.

\section{Adaptive double-phase ROF model}

In order to reduce the staircasing effect of the classical ROF model, we propose the following approach. 
Given a noisy image $g\colon  \Omega \to [0,1]$ and a positive number $\lambda_0$, we first solve the classical ROF minimization problem
\begin{equation}
\min_u \int_{\Omega} |\nabla u| \, \dd x + \frac{1}{2 \lambda_0} \int_\Omega |u-g|^2 \,\dd x .
\end{equation}  
The minimized functional is strictly convex, lower semicontinuous and coercive, and the minimization problem yields a unique solution $u_{\rm ROF}$. From the solution of the classical ROF problem, we construct a nonnegative weight in the following way. Let $W\colon [0,\infty) \to [0,\infty)$ be a nonincreasing function with compact support and $W(0) > 0$. Denote by $\tilde{u}_{\rm ROF}$ the convolution of the solution $u_{\rm ROF}$ with a mollifying kernel $\varrho_r$. For simplicity, we choose $\varrho_r$ to be a uniform density on the ball $B_r$  with a given radius $r > 0$. We then set
\begin{equation}\label{eq:weight}
w(x) = W(|\nabla \tilde{u}_{\rm ROF}|(x)).
\end{equation}
Having constructed the weight, we consider the double-phase ROF minimization problem
\begin{equation}\label{eq:dprof}
\min_u \int_{\Omega} |\nabla u| \, \dd x + \int_\Omega w(x) |\nabla u|^2 \, \dd x
+ \frac{1}{2\lambda} \int_\Omega |u-g|^2 \,\dd x.
\end{equation} 
The well-posedness of this type of problems was considered e.g. in \cite{HarjulehtoHasto2021} and the corresponding Euler--Lagrange equation in the general case of the variable-growth total variation $\TV_\varphi$ was identified in \cite{GornyLasicaMatsoukas2025}. The functional appearing in \eqref{eq:dprof} is also strictly convex, lower semicontinuous and coercive, and thus admits a unique solution which we denote $u_{\rm dpROF}$. Solving in the first step the classical ROF problem serves essentially as an edge detection method: the leading idea behind the above scheme is that in the areas close to the jumps in the original image, we want to process the image using the $1$-Laplacian operator in order to ensure that the edges are well-preserved, while far from the jumps of the original image we add a Laplacian term with a large weight to reduce staircasing. We summarize the described algorithm below.

\begin{algorithm}\label{alg:adaptivedoublephase}
\caption{The adaptive double-phase ROF model}
\begin{algorithmic}[1]
\STATE Fix a noisy image $g$ and choose a parameter $\lambda_0 > 0$ suitable for the classical ROF minimization problem \eqref{eq:rof}.

\STATE Find a solution $u_{\rm ROF}$ of the classical ROF minimization problem for the datum $g$ and~the parameter $\lambda_0$.

\STATE Choose a mollification radius $r > 0$ and a non-increasing function $W$ with compact support and $W(0) > 0$.

\STATE Mollify the function $u_{\rm ROF}$ with a kernel of radius $r$; denote the result by $\tilde{u}_{\rm ROF}$.

\STATE Compute the weight $w$ using the formula \eqref{eq:weight} from the mollified function $\tilde{u}_{\rm ROF}$.

\STATE Choose a parameter $\lambda >0$. Find a solution $u_{\rm dpROF}$ of the double-phase ROF minimization problem \eqref{eq:dprof} for the datum $g$, the parameter $\lambda$ and the weight $w$.
\end{algorithmic}
\end{algorithm}

Let us note that there are several free parameters in the above algorithm: coefficients $\lambda_0$ and $\lambda$, mollification radius $r$ and weight $W$. The simple choice $\lambda_0 = \lambda$ gives satisfactory results, but it can be improved upon (see Table \ref{tab:dprof_par}). The coefficient $\lambda$ (as in the classical ROF model \eqref{eq:rof}) is used to balance the regularization effect of the (generalized) total variation term and the accuracy of the reconstruction ensured by the fidelity term; too large values of $\lambda$ result in flattening of the image, while too small values of $\lambda$ result in small reduction of noise. The mollification radius $r$ should be small 
and ensures that the support of the weight is away from the (sufficiently large) jumps of $u_{\rm ROF}$. In preliminary tests, the value of 2 pixels turned out optimal in most cases, and we will always use it in experiments with natural raster images. 

The last free parameter is the choice of the weight $W$. It needs to vanish for large arguments 
so that the support of the resulting weight $w$ is away from the (sufficiently large) jumps of $u_{\rm ROF}$. The value $W(0)$ and the shape of the function 
determine how strong is the effect of the regularization with the Laplacian with respect to the $1$-Laplacian regularization away from the jumps of $u_{\rm ROF}$. We investigated different choices of~$W$ (see github repository \cite{github} for detailed report on this and other experiments), and concluded that the performance of the model does not change much depending on its shape, and the piecewise linear function with two free parameters $a, b >0$ given by 
\[W(x) = \max \left(0, a - b \max \left(x, a/(2b) \right) \right)\]
gives good results and has enough flexibility. This is the choice we use in experiments throughout the paper. 

\begin{remark}
We note that an adaptive denoising model based on a regularizer of form similar to \eqref{dpint} was proposed in \cite{LiuTanSu}. In that paper, apparently a weight of form $W(|\nabla f|)$ is used. The performance of the proposed model is compared against several denoising models, and numbers showing superiority of the proposed model are given. However, as the authors give no details concerning the parameters etc.\ chosen for the comparison models, it is difficult to reproduce their results. Moreover, the weight they consider is non-vanishing, so their regularizer is not a double-phase (variable growth) integral.
\end{remark}

\section{Numerical implementation}\label{sec:chambollepock}

In order to numerically solve the minimization problems constituting our model, we will use the Chambolle--Pock algorithm, first introduced in \cite{ChambollePock}. It is a first order primal-dual algorithm for solving (possibly non-smooth) convex optimization problems. In the simplest form, it can be shortly described as follows. Let $X$ and $Y$ be finite-dimensional real vector spaces and let $K\colon X \rightarrow Y$ be a bounded linear operator. Then, for proper, convex and lower semicontinuous functionals $F\colon Y \rightarrow [0,\infty]$ and $G\colon X \rightarrow [0,\infty]$, we consider  a minimization problem of the form
\begin{equation}
\min_{x \in X} F(Kx) + G(x).
\end{equation}
A primal-dual formulation of this problem takes the form
\begin{equation}\label{eq:chambollepockprimaldual}
\min_{x \in X} \max_{y \in Y} \, \langle Kx, y \rangle_Y + G(x) - F^*(y),
\end{equation}
where $\langle \cdot, \cdot \rangle$ denotes the scalar product and $F^*$ is the convex conjugate of $F$. The Chambolle--Pock algorithm which recovers a solution $(\hat{x},\hat{y}) \in X \times Y$ of the primal-dual formulation, and consequently the solution $\hat{x} \in X$ of the original minimisation problem, is the following.

\begin{algorithm}
\caption{The Chambolle--Pock algorithm}
\begin{algorithmic}[1]
\STATE Initialization: Fix $\tau, \sigma > 0$, $\theta \in [0,1]$ and initial data $(x^0,y^0) \in X \times Y$. Set $\bar{x}^0 = x^0$.

\STATE Iterations: For $n \geq 0$, update the variables $x^n, y^n, \bar{x}^n$ as follows:
\begin{equation}
y^{n+1} = (I + \sigma \partial F^*)^{-1} (y^n + \sigma K \bar{x}^n);
\end{equation}
\begin{equation}
x^{n+1} = (I + \tau \partial G)^{-1} (x^n - \tau K^* y^{n+1});
\end{equation}
\begin{equation}
\bar{x}^{n+1} = x^{n+1} + \theta(x^{n+1} - x^n).
\end{equation}
\end{algorithmic}
\end{algorithm}

The convergence of the algorithm (with the convergence rate $O(1/N)$) in the case $\theta = 1$ is guaranteed by \cite{ChambollePock}. 
The case $\theta = 0$ corresponds to the standard Arrow--Hurwicz algorithm \cite{ArrowHurwicz} and convergence results are also well-known. 
If the functional $G$ is uniformly convex as will be in our case, the Chambolle--Pock algorithm may be modified to obtain an improved convergence rate $O(1/N^2)$. 
This is done by adjusting in each iteration the step variables $\tau$, $\sigma$ and $\theta$ which were constant in the original formulation. In the present work, we primarily use the accelerated version of the Chambolle--Pock algorithm.

The main difficulty in applying the algorithm in a particular setting lies in finding the proximal operators $(I + \sigma \partial F^*)^{-1}$ and $(I + \tau \partial G)^{-1} (x^n - \tau K^* y^{n+1})$. 


Let us now recall how to present the classical ROF model in the framework of the Chambolle--Pock algorithm (see \cite{ChambollePock}). We represent an image as a function on a regular Cartesian grid of size $M \times N$.
The (finite dimensional) vector spaces are $X = \mathbb{R}^{MN}$ and $Y = X \times X
$ equipped with the standard scalar product. The linear operator $K\colon  X \rightarrow Y$ is the (discrete) gradient: we set $K = \nabla$ where $\nabla$ is defined for $u \in X$ by
\begin{equation}\label{eq:definitionofgradient}
(\nabla u)_{i,j} = ( u_{i+1,j} - u_{i,j},  u_{i,j+1} - u_{i,j})^T 
\end{equation}
with zeros for $i = M$ and $j = N$ respectively. With this definition, the gradient is a continuous linear operator with norm $\| \nabla \| \leq 8$, see \cite{ChambollePock}. The dual $K^*$ is the minus divergence operator.

The ROF model \eqref{eq:rof} corresponds to the following choices: 
the functional $F\colon Y \rightarrow \mathbb{R}$ is given by $F(p) =  \| p \|_1$, and the functional $G\colon X \rightarrow \mathbb{R}$ is defined as~$G(u) = \frac{1}{2 \lambda} \| u - g \|_2^2$. The resolvent operators $(I + \sigma \partial F^*)^{-1}$ and $(I + \tau \partial G)^{-1}$ take the following form: we have
$$p = (I+\sigma \partial F^{*})^{-1} (\tilde{p}) \quad \equiv \quad  p_{i,j} = \frac{\tilde{p}_{i,j}}{\max (1,|\tilde{p}_{i,j}|)}$$
and
\begin{equation}\label{eq:resolventforg}
u = (I + \tau \partial G)^{-1}(\widetilde{u}) \quad \equiv \quad u_{i,j} = \frac{\widetilde{u}_{i,j} + \frac\tau\lambda g_{i,j}}{1 + \frac\tau\lambda}
\end{equation}
(note that $\lambda$ in \cite{ChambollePock} corresponds to $\frac{1}{\lambda}$ in our paper).

Similarly, the Huber-ROF model
corresponds to 
$F\colon Y \rightarrow \mathbb{R}$ given by $F(p) = \| | p |_\alpha \|_1$ with $|p|_\alpha$ given by \eqref{huber_func}.
In this case, the proximal operator $(I+\sigma \partial F^{*})^{-1}$ is given by
$$p = (I+\sigma \partial F^{*})^{-1} (\tilde{p}) \quad \equiv \quad  p_{i,j} = \frac{\frac{\tilde{p}_{i,j}}{1 + \sigma \alpha}}{\max (1,|\frac{\tilde{p}_{i,j}}{1+\sigma\alpha}|)}.$$

In the case of minimization problem for the double-phase ROF functional \eqref{eq:dprof}, $F$ is of the form 
\begin{equation} \label{eq:Fvarphi}
F(p) = \sum_{i,j} \varphi((i,j),|p_{i,j}|), 
\end{equation}
with 
\begin{equation}\label{eq:doublephasevarphi}
\varphi((i,j),|p_{i,j}|) = |p_{i,j}| + w_{i,j} |p_{i,j}|^2.
\end{equation}
In order to apply the Chambolle--Pock algorithm, we need to identify the resolvent operator $(I + \sigma \partial F^*)^{-1}$. 

\begin{proposition}
Let $F$ be given by \eqref{eq:doublephasevarphi}, \eqref{eq:Fvarphi}. 
It holds that
\begin{align}
&p=(I+\sigma \partial F^{*})^{-1}(\tilde{p}) \Longleftrightarrow \\
& p_{i,j} = \begin{cases}
\frac{\tilde{p}_{i,j}}{\max (1,|\tilde{p}_{i,j}|)} &  \text{if } w_{i,j} = 0; \\
\displaystyle  \min \left(1, \frac{w_{i,j}|\tilde{p}_{i,j}|+\sigma}{w_{i,j} |\tilde{p}_{i,j}| + \sigma |\tilde{p}_{i,j}|}\right) \tilde{p}_{i,j} & \text{if } w_{i,j} > 0.
\end{cases}
\end{align}
\end{proposition}

\begin{proof}
Let us first observe that the conjugate functional $F^*: Y \rightarrow \mathbb{R}$ is given by
$$F^{*}(p) = \sum_{i,j} \varphi^{*}((i,j),|p_{i,j}|),$$ 
where $\varphi^*((i,j),|p_{i,j}|)$ is defined pointwise as
$$
\varphi^*((i,j),|p_{i,j}|) =
\begin{cases}
\infty \chi_{(1,\infty)}(|p_{i,j}|) & \text{if } w_{i,j} = 0; \\
\displaystyle  \frac{\max\{0,|p_{i,j}| - 1\}^{2}}{2 w_{i,j}} & \text{if } w_{i,j} > 0.
\end{cases}
$$
Thus, the functional $F^{*}$ takes the following form
\begin{align}
F^{*}(p) = \sum_{i,j} \infty &\chi_{\{w=0\}}((i,j))\, \chi_{(1,\infty)}(|p_{i,j}|) \\ &+ \chi_{\{w>0\}}((i,j))\, \frac{\max\{0,|p_{i,j}| - 1\}^{2}}{2 w_{i,j}}.
\end{align}
In the case $w_{i,j} = 0$, $\varphi^*((i,j),|p_{i,j}|)$ is the indicator function of the unit ball, and the resolvent operator is given by pointwise projections onto $L^2$ unit balls, that is 
$$p=(I+\sigma \partial F^{*})^{-1} (\tilde{p}) \Longleftrightarrow  p_{i,j} = 
\frac{\tilde{p}_{i,j}}{\max (1,|\tilde{p}_{i,j}|)}.$$
For $w_{i,j} > 0$, the function $\varphi^{*}((i,j),|p_{i,j}|)$ is differentiable in the second coordinate with the gradient given by 
$$ \nabla_p \varphi^*((i,j),|p_{i,j}|) =
\begin{cases}
0 & \text{if } |p_{i,j}| \le 1; \\
\displaystyle  \frac{(|p_{i,j}|-1)}{w_{i,j} |p_{i,j}|} p_{i,j} & \text{if } |p_{i,j}| > 1.
\end{cases}
$$
To compute the resolvent $p=(I+\sigma \partial F^{*})^{-1}(\tilde{p})$, we now solve
$$\tilde{p}_{i,j} = p_{i,j} + \sigma \nabla_{p} \varphi^*(x,p_{i,j}).$$

If $|p_{i,j}|\le 1$, then $\nabla_{p} \varphi(x,p_{i,j})=0$ and so $$ p=(I+\sigma \partial F^{*})^{-1}(\tilde{p}) \Longleftrightarrow p_{i,j} = \tilde{p}_{i,j}.$$
For $|p_{i,j}|>1$, using the expression for the gradient we get
\begin{equation}
p=(I+\sigma \partial F^{*})^{-1}(\tilde{p}) \Longleftrightarrow \tilde{p}_{i,j} = \left(1+\sigma \frac{(|p_{i,j}|-1)}{w_{i,j} |p_{i,j}|}  \right) p_{i,j}
\end{equation}
or equivalently
\begin{equation}
p_{i,j} = \lambda_{i,j} \tilde{p}_{i,j} \, \text{ with } \, \lambda_{i,j} = \left( \frac{w_{i,j} |\tilde{p}_{i,j} | + \sigma}{w_{i,j} |\tilde{p}_{i,j}| + \sigma |\tilde{p}_{i,j}|}\right).
\end{equation}

Thus, the final expression for the resolvent is
\begin{equation}
p_{i,j} = \begin{cases}
\frac{\tilde{p}_{i,j}}{\max (1,|\tilde{p}_{i,j}|)} &  \text{if } w_{i,j} = 0; \\
\displaystyle  \min \left(1, \frac{w_{i,j}|\tilde{p}_{i,j}|+\sigma}{w_{i,j} |\tilde{p}_{i,j}|+\sigma |\tilde{p}_{i,j}|}\right) \tilde{p}_{i,j} & \text{if } w_{i,j} > 0,
\end{cases}
\end{equation}
which concludes the proof.
\end{proof}

\section{Qualitative experiments}
\label{sec:qual}

We first tested how the adaptive double-phase ROF model performs in a model one-dimensional case. Our goal is to reduce staircasing with respect to the classical ROF model, while remaining reasonably accurate to the original picture.

In Figures \ref{fig:sawpicture} and \ref{fig:cut1picture} we present visually the effect of the adaptive double-phase ROF model on two test images, for several choices of parameters for the weight and in comparison with the results of the classical ROF model and the Huber-ROF model for two choices of parameters. The parameters $\lambda_0$ and~$\lambda$ are set to $0.24$. 
The function in Figure \ref{fig:sawpicture} 'saw' is a synthetic datum with several jumps of different sizes and a linear interpolation between them, while the function in Figure \ref{fig:cut1picture} 'cut1' represents a one-dimensional slice (a single row) of the natural image 'schnitzel' (see below). We have used a rescaled version of 'cut1' in Figure \ref{fig:cut1picture}. 


We observe the following effects. In both Figure \ref{fig:sawpicture} and Figure \ref{fig:cut1picture}, we observe that the staircasing which is clearly visible in the reconstruction using the classical ROF model is reduced, especially for higher values of the parameters.
The Huber-ROF model is also capable of reducing staircaising somewhat. 
However it also retains spurious wavy structures, which are not present in the adaptive double-phase ROF model. 

\begin{figure}[H]
    \centering
    \includegraphics[width=0.96\textwidth]{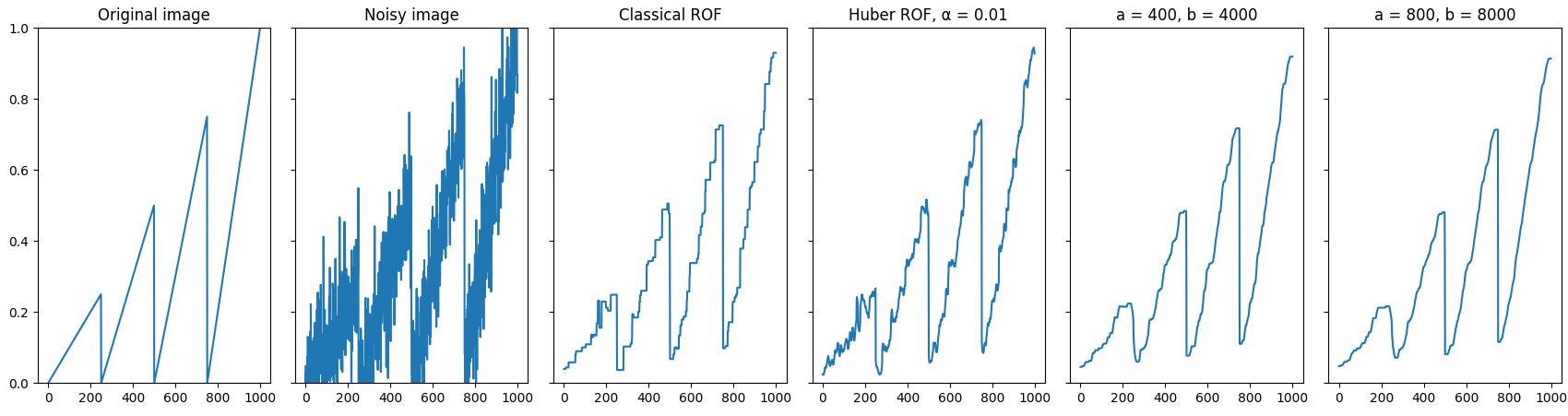}
    \vspace{-8pt} 
    \caption{\small Reconstruction of a one-dimensional synthetic image with classical ROF, Huber-ROF, and adaptive double-phase ROF (with two choices of parameters).}
    \label{fig:sawpicture} 

    \vspace{10pt}
    \centering
    \includegraphics[width=0.96\textwidth]{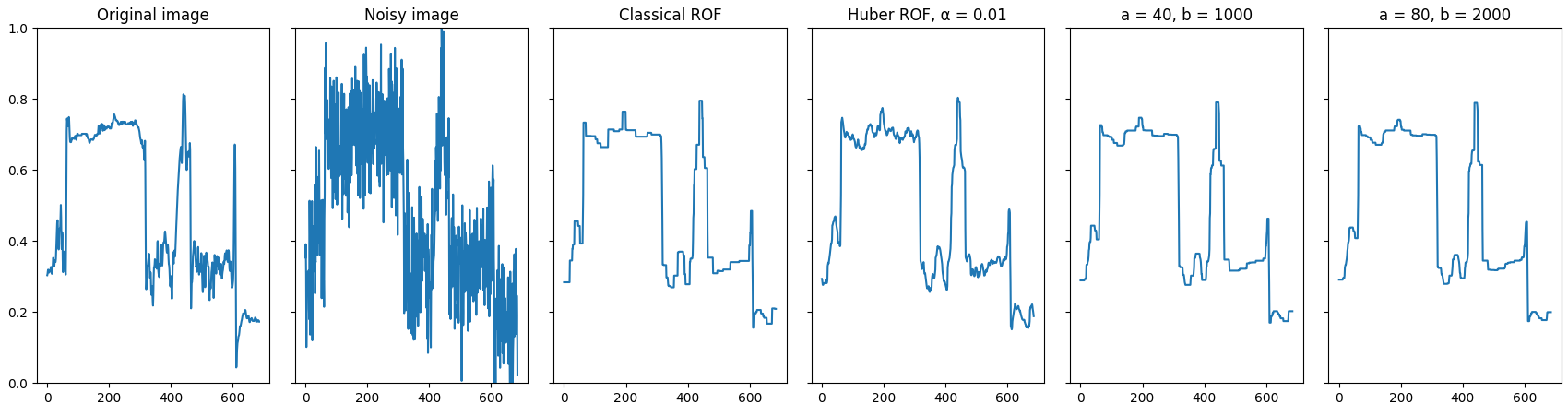}
    \vspace{-8pt} 
    \caption{\small Reconstruction of a one-dimensional natural image with classical ROF, Huber-ROF, and adaptive double-phase ROF (with two choices of parameters).}
    \label{fig:cut1picture}

    \vspace{10pt}
    \centering
    \includegraphics[width=0.96\linewidth]{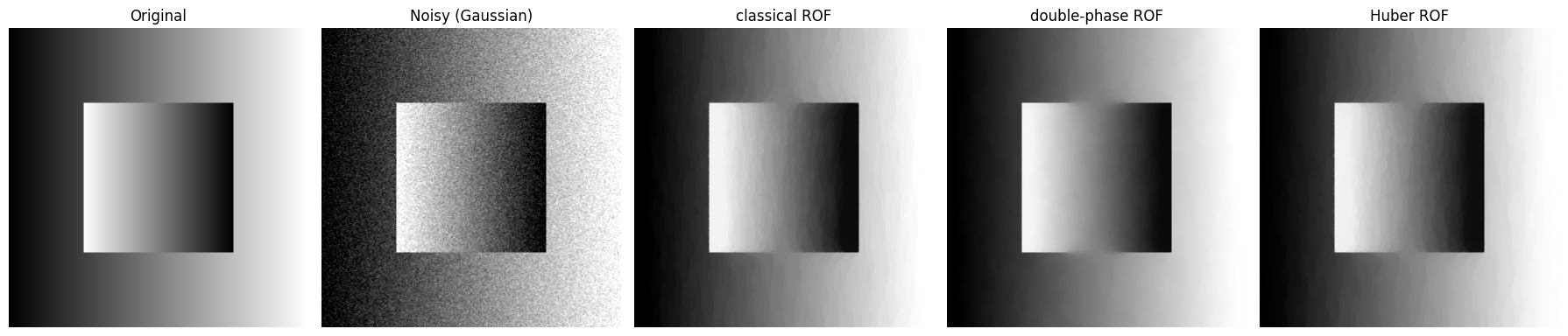}
    \vspace{-8pt} 
    \caption{\small Original 'double gradient' synthetic image; image with added Gaussian noise ($\sigma^2=0.01$); denoising results corresponding to the maximum SSIM values, respectively: classical ROF ($\lambda = 0.24$, ${\rm SSIM} = 0.926$, ${\rm PSNR} = 32.81$), double-phase ROF ($\lambda = 0.12$, ${\rm SSIM} = 0.962$, ${\rm PSNR} = 34.38$), and Huber-ROF with $\alpha = 0.01$ ($\lambda = 0.26$, ${\rm SSIM} = 0.963$, ${\rm PSNR} = 34.02$). }
    \label{fig:denoiseddg2_256}
\end{figure}

\begin{figure}[H]
    \centering
    \includegraphics[width=0.96\linewidth]{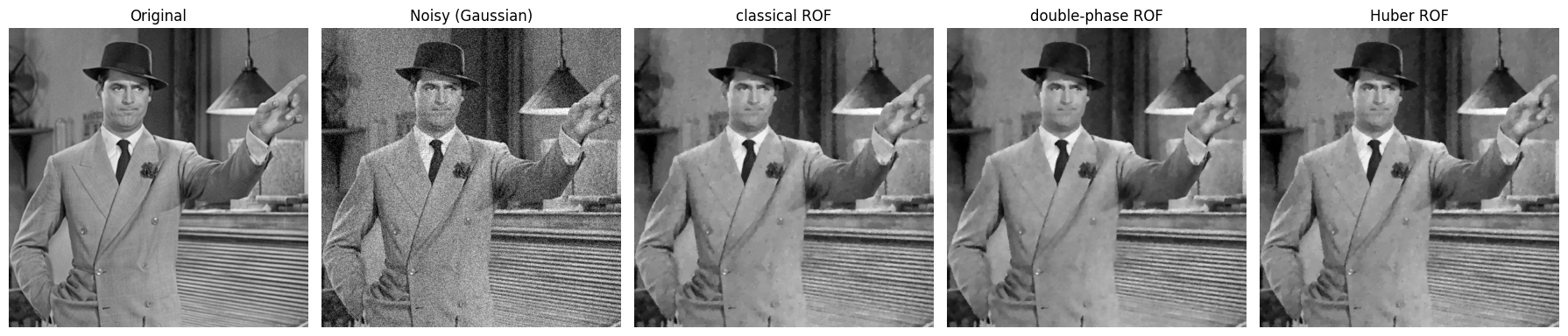}
    \vspace{-8pt}
    \caption{\small Original 'carygrant' image; image with added gaussian noise ($\sigma^2=0.01$); denoising results corresponding to the maximum SSIM values, respectively: classical ROF ($\lambda = 0.06$), double-phase ROF ($\lambda = 0.08$), and Huber ROF ($\alpha=0.01, \lambda = 0.08$). }
    \label{fig:denoisedCG1}

    \centering
    \vspace{10pt} 
    \includegraphics[width=0.96\linewidth]{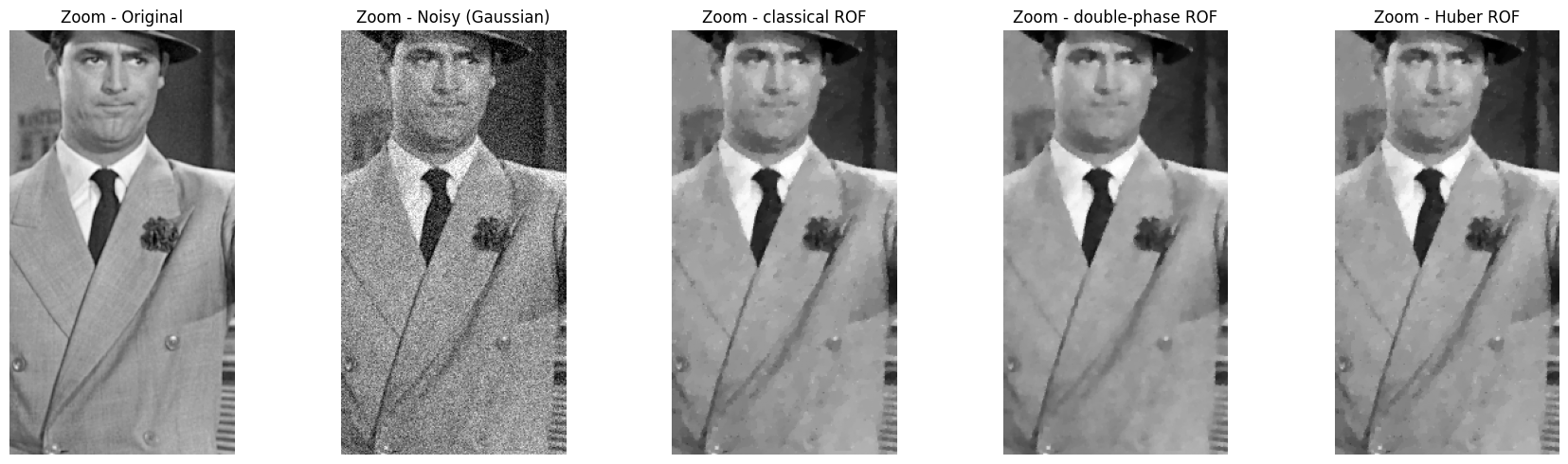}
    \vspace{-8pt}
    \caption{\small Magnified part of Figure \ref{fig:denoisedCG1} 'carygrant': face and suit.}
    \label{fig:denoisedCGZOOM1a}

    \centering
    \vspace{10pt}
    \includegraphics[width=0.96\linewidth]{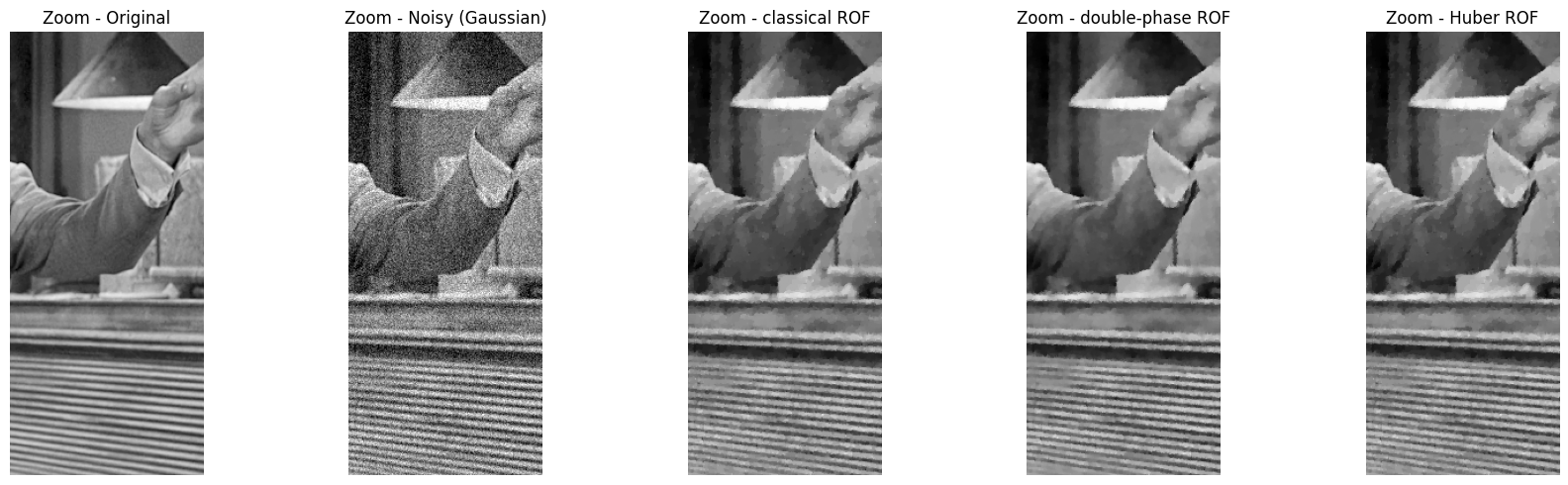}
    \vspace{-8pt}
    \caption{\small Magnified part of Figure \ref{fig:denoisedCG1} 'carygrant': arm and lamp.}
    \label{fig:denoisedCGZOOM1b}
  \centering
  \vspace{10pt}
  \subfloat[Mollified gradient or $u_{\rm ROF}$]{
  \includegraphics[width=0.32\textwidth]{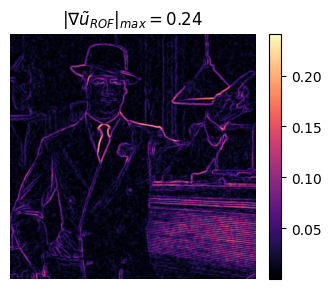}}
  \subfloat[Dependence on $|\nabla \tilde{u}_{\rm ROF}|$.]{\includegraphics[width=0.32\textwidth]{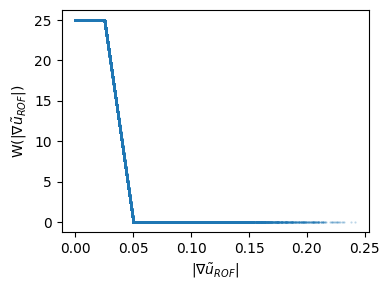}}
  \subfloat[The weight.]{\includegraphics[width=0.32\textwidth]{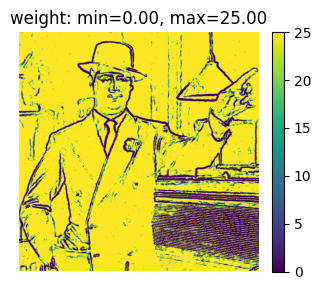}}
  \caption{\small Construction of the weight from mollified gradient of $u_{\rm ROF}$ with $\lambda=0.06$ for the 'carygrant' image.}
  \label{fig:weightCG1}
\end{figure}

\begin{figure}[H]
    \centering
    \includegraphics[width=0.96\linewidth]{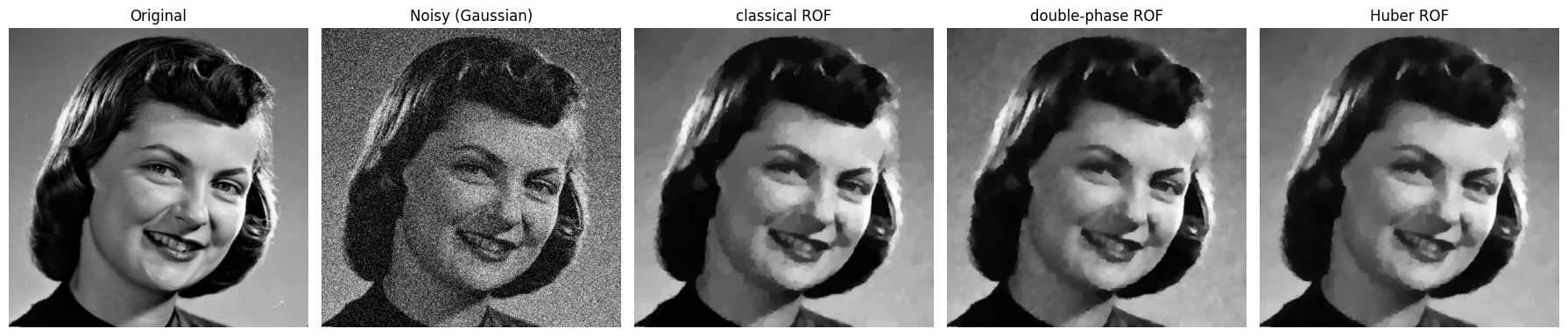}
    \vspace{-8pt}
    \caption{\small Original 'girlface' image; image with added gaussian noise ($\sigma^2=0.04$); denoising results corresponding to the maximum SSIM values, respectively: classical ROF ($\lambda = 0.22$), double-phase ROF ($\lambda = 0.14$), and Huber ROF ($\alpha=0.01$, $\lambda = 0.24$).}
    \label{fig:denoisedGF2}

    \vspace{10pt}
    \centering
    \includegraphics[width=0.96\linewidth]{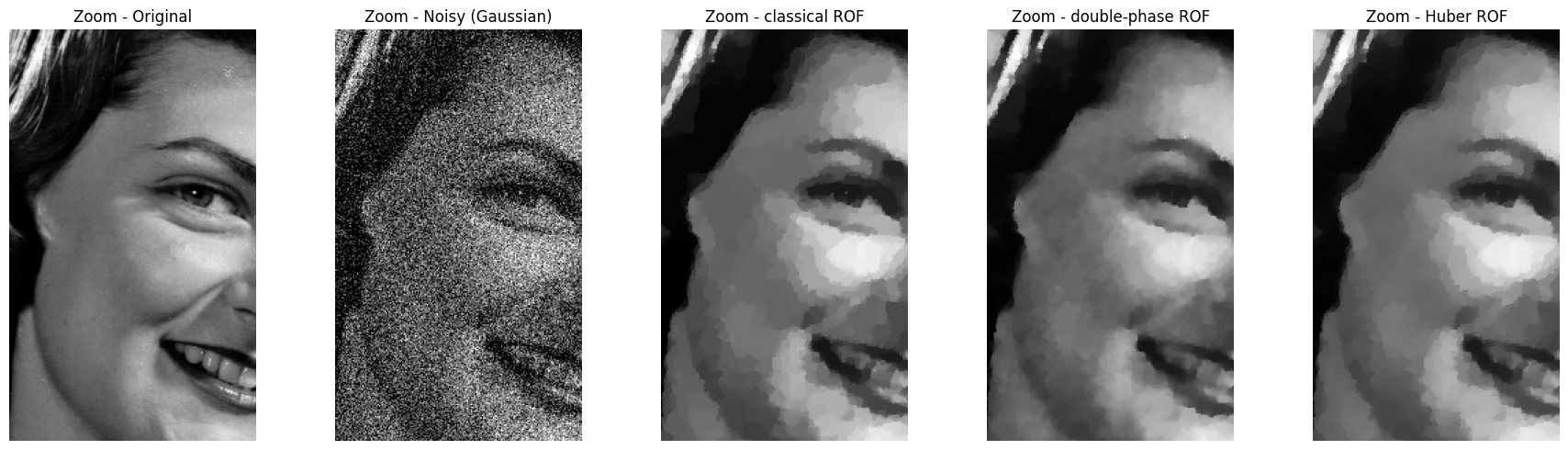}
    \vspace{-8pt}
    \caption{\small Magnified part of Figure \ref{fig:denoisedGF2} 'girlface': eye and cheek.}
    \label{fig:denoisedGFZOOM2}
\end{figure}

\begin{figure}[H]
    \centering
    \includegraphics[width=0.96\linewidth]{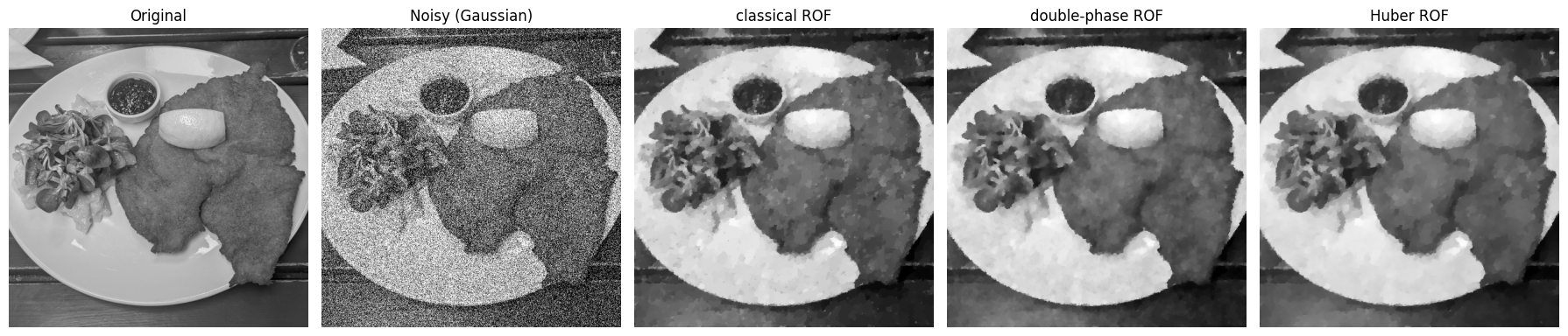}
    \vspace{-8pt}
    \caption{\small Original 'schnitzel' image; image with added Gaussian noise ($\sigma^2=0.07$); denoising results corresponding to the maximum SSIM values, respectively: classical ROF ($\lambda = 0.22$, ${\rm SSIM} = 0.555$, ${\rm PSNR} = 21.58$), double-phase ROF ($\lambda = 0.16$, ${\rm SSIM} = 0.576$, ${\rm PSNR} = 22.15$), and Huber ROF with $\alpha = 0.01$ ($\lambda = 0.26$, ${\rm SSIM} = 0.566$, ${\rm PSNR} = 22.02$).}
    \label{fig:denoisedschn4_512}

    \centering
    \vspace{10pt}
    \includegraphics[width=0.96\linewidth]{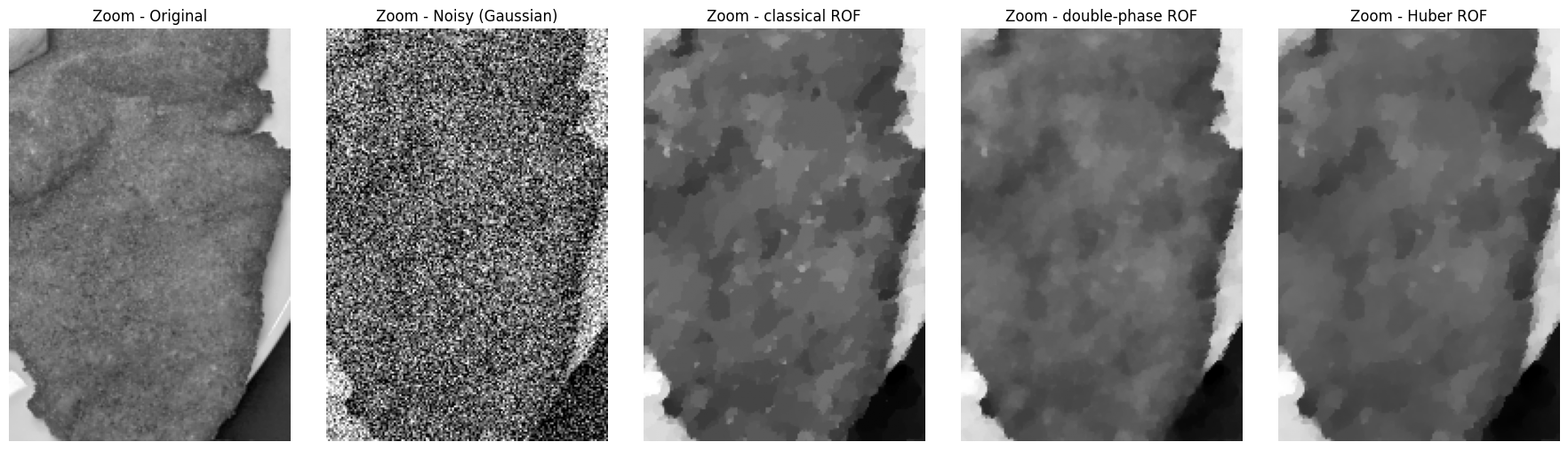}
    \vspace{-8pt}
    \caption{\small Magnified part of Figure \ref{fig:denoisedschn4_512} 'schnitzel': breadcrumbs texture.}
    \label{fig:schnitzelzoom}
\end{figure}

\newpage

We next present the results of several numerical experiments concerning the performance of the adaptive double-phase ROF model for two-dimensional images. As in the one-dimensional case, our goal is to reduce the staircasing effect of the classical ROF model, and we use a similar setup. To quantify the accuracy of the reconstruction, we 
used the structural similarity index measure (SSIM) \cite{WBSS2004} and the peak signal-to-noise ratio (PSNR), see e.g. \cite{Salomon2007}. However, in this 
section we focus on presenting the visual results, showing the graphs for the respective measures only on a few occasions. For a detailed presentation of numerical experiments, including the calculations of these metrics for various choices of parameters, we refer to the available code and summary files \cite{github}.  In the presentation below, in each Figure with images, from left to right we show: the original image; the image after adding a Gaussian noise of prescribed variance; the result of the reconstruction using the classical ROF model; the adaptive double-phase ROF model; and the Huber-ROF model with $\alpha = 0.01$.

Let us first present visually the result for a synthetic image (in Figure \ref{fig:denoiseddg2_256}) for the noise level $\sigma^2 = 0.01$; note that all our images are normalized so that they take values in the interval $[0,1]$. We observe that, in comparison with the classical ROF model and the Huber-ROF model, when using the adaptive double-phase ROF model there is much less staircasing in the inside square; additionally, the gradient in the outer square is more smoothed out. We also observe that the edges of the inner square are well-preserved, with the exception of the area in the middle when they are small (this effect appears also in the Huber-ROF model and, to a lesser degree, in the classical ROF model).

Then, we present the results for several classical test images. First, in Figure \ref{fig:denoisedCG1}, we consider the 'carygrant' image for the noise level $\sigma^2 = 0.01$. We use parameters $a = 50$ and $b = 1000$ for the weight. We again observe that the main features of the original image are well-preserved and that there is much less staircasing, especially on the cheeks, the suit and the hand; enlarged parts of these pictures for comparison are shown in Figures \ref{fig:denoisedCGZOOM1a}-\ref{fig:denoisedCGZOOM1b}. Furthermore, the texture visible at the bottom of Figure \ref{fig:denoisedCGZOOM1b} is preserved better. The construction of the weight is presented in Figure \ref{fig:weightCG1}: we first mollify the gradient of the solution to the ROF problem (on the left-hand side), we apply to it the weight visible in the middle picture, and the resulting weight is presented on the right-hand side.

In Figure \ref{fig:denoisedGF2}, we present the same set of images for the 'girlface' test image for a higher noise level $\sigma^2 = 0.04$. We use parameters $a = 60$ and $b = 1200$ for the weight. The quality of the reconstructed image is naturally slightly worse than in the previous example (for each model), but the image is reasonably well-preserved and again we observe less staircasing with respect to the classical ROF model and the Huber-ROF model, in particular in the area around the eye and on the cheek; enlarged parts of these pictures are presented in Figure \ref{fig:denoisedGFZOOM2}. Moreover, the changes in intensity of the picture are more gradual with respect to the other two models.

In Figure \ref{fig:denoisedschn4_512}, we present the same set of images for a test image of a schnitzel, again for a higher noise level $\sigma^2 = 0.07$. We use parameters $a = 60$ and $b = 1200$ for the weight. This time the reconstructed image is visibly blurred (for each model), but the overall contours are well-preserved and again we observe less staircasing with respect to the classical ROF model and the Huber-ROF model, in particular on the schnitzel itself; magnified parts of these pictures are presented in Figure~\ref{fig:schnitzelzoom}. We believe the double-phase ROF model was capable of reproducing the crispy texture of the schnitzel more naturally than the other models.

\section{Comparison of optimised models\label{sec:comparison}}

In this section, we report on a quantitative comparison of the performance of dpROF against several established, classical denoising methods with similar interpretability, leaving out neural-network based models. We carried out a performance comparison experiment using the Berkeley Segmentation Data Set (BSDS500). For noise variance values 
\[\sigma^2 \in \{0.001, 0.002, 0.004, 0.007, 0.01, 0.02, 0.04, 0.07\}\] 
we found optimal values of the parameters in the competing models on the training set. 
(see the 'comparison' supplement in \cite{github} for details). We chose the optimal values based on the average distance of denoised images from original images in the SSIM metric. 

Based on optimization experiments, we fix a variant of our model with prescribed values of parameters (given in Table \ref{tab:dprof_par}) depending on the noise variance, which can be approximately determined for a given image using scikit-image function {\tt estimate\_sigma}.  

\begin{table}[h] 
    \centering
    \renewcommand{\arraystretch}{1.5}
    \caption{Choice of parameters in the adaptive double-phase ROF model}
    \begin{tabular}{c|cccc}
        \hline
        Noise level & $\lambda_{\rm ROF}$ & $\lambda$ & $a$ & $b/a$\\
        \hline
        Small ($\sigma^2 \leq 0.005$) & \multirow{2}{*}{$0.7 \cdot \sigma$} & \multirow{3}{*}{$0.5 \cdot \sigma$} & \multirow{3}{*}{60} & 100 \\
        Medium ($0.005 < \sigma^2 \leq 0.015$) & & & & 50 \\
        Large ($\sigma^2 > 0.015$) &  $0.8 \cdot \sigma$ & & & 25 \\
        \hline
        \multicolumn{4}{c}{\vspace{0.5em}} \\ 
    \end{tabular}
    \label{tab:dprof_par}
\end{table}

In addition to ROF and Huber-ROF, we compare dpROF to more recent, established models: Chen--Levine--Rao (CLR) variable exponent model \cite{CLR}, Total Generalized Variation (TGV) \cite{BKP2010, KBS2011} and non-local means (NLM) \cite{BCM2011}. CLR is an adaptive variant of ROF, where $|\nabla u|$ replaced by a variable exponent integrand in areas where sharp contours are not expected. To be precise, it amounts to the following minimization problem: 
\begin{equation}\label{clr}  \min_{u}\  \int \varphi(x,\nabla u)\, \dd x + \frac{1}{2\lambda}\int|u-g|^2 \,\dd x,   
\end{equation} 
where 
\begin{equation}  \label{clrphi} 
\varphi(x,\xi) = 
\begin{cases} 
\frac{1}{p(x)} |\xi|^{p(x)} & \text{if } |\xi| \leq \beta, \\
|\xi| - \frac{\beta p(x) - \beta^{p(x)}}{p(x)} & \text{if } |\xi| > \beta,
\end{cases}
\end{equation}
and the exponent $p(x)$ is given by 
\[ p(x) = 1 + (1 + k | \nabla G_\sigma * g(x)|^2)^{-1}. \]
In the above formulae,  $G_\sigma(x) = \sigma^{-1} \exp(- |x|^2/2 \sigma^2)$ is the Gaussian kernel and $\lambda >0$, $\beta >0$, $k >0$, $\sigma >0$ are parameters. An important qualitative difference from our model is that in CLR, $\varphi(x, |\xi|$) is always linear in $|\xi|$ for large values of $|\xi|$. It is also difficult to solve efficiently by the Chambolle--Pock scheme, as the resolvent is not given by a simple explicit formula. Following the authors \cite{CLR}, we implemented a numerical scheme based on an explicit discretization of the gradient flow of the minimized quantity. 

\begin{figure}[H]
    \centering
    \includegraphics[width=0.9\linewidth]{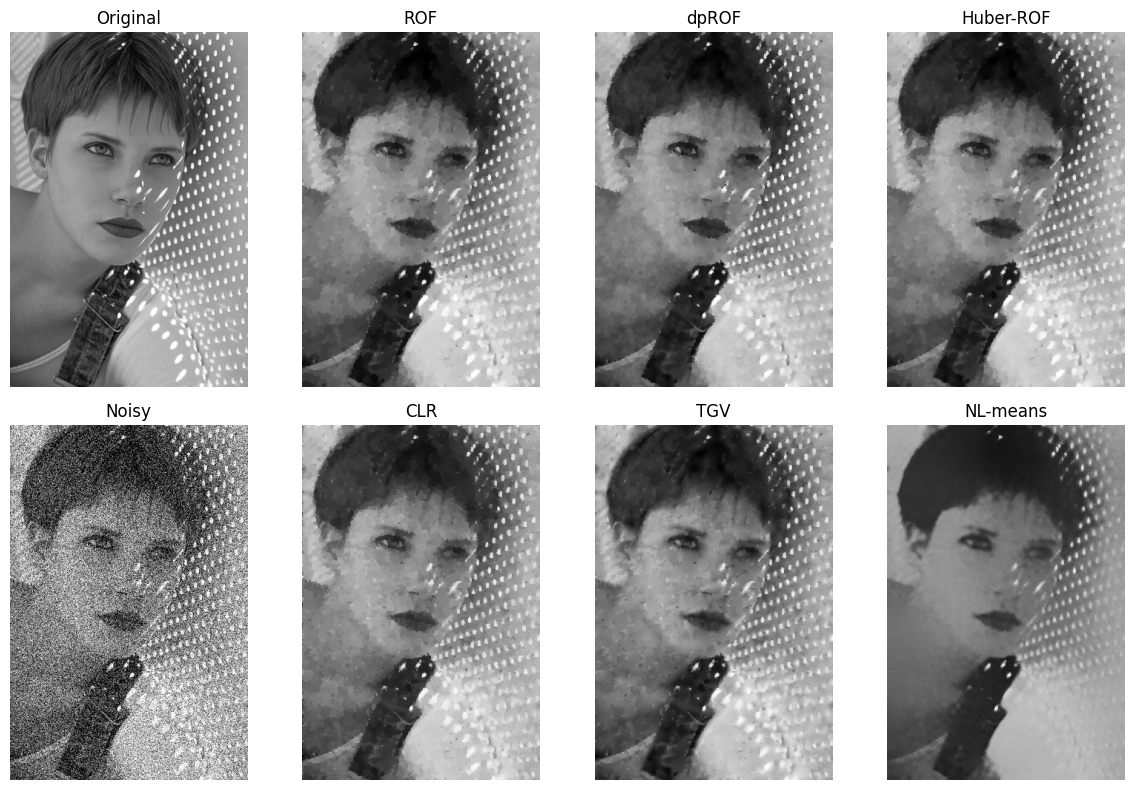}
    \vspace{-8pt}
    \caption{\small Original 'lady' image; image with added Gaussian noise with $\sigma^2=0.04$; optimized denoising results, respectively: ROF (${\rm SSIM} = 0.669$, ${\rm LPIPS} = 0.291$, $t = 15.1$ s), dpROF (${\rm SSIM} = 0.690$, ${\rm LPIPS} = 0.273$, $t = 0.7 + 6.7$ s), Huber-ROF (${\rm SSIM} = 0.675$, ${\rm LPIPS} = 0.284$, $t = 14.7$ s), CLR (${\rm SSIM} = 0.655$, ${\rm LPIPS} = 0.326$, $t = 97.2$ s), TGV (${\rm SSIM} = 0.663$, ${\rm LPIPS} = 0.304$, $t = 40.9$ s), NLM (${\rm SSIM} = 0.607$, ${\rm LPIPS} = 0.327$, $t = 11.4$ s).}
    \label{fig:lady_denoise}

    \vspace{10pt} 
    \includegraphics[width=0.9\linewidth]{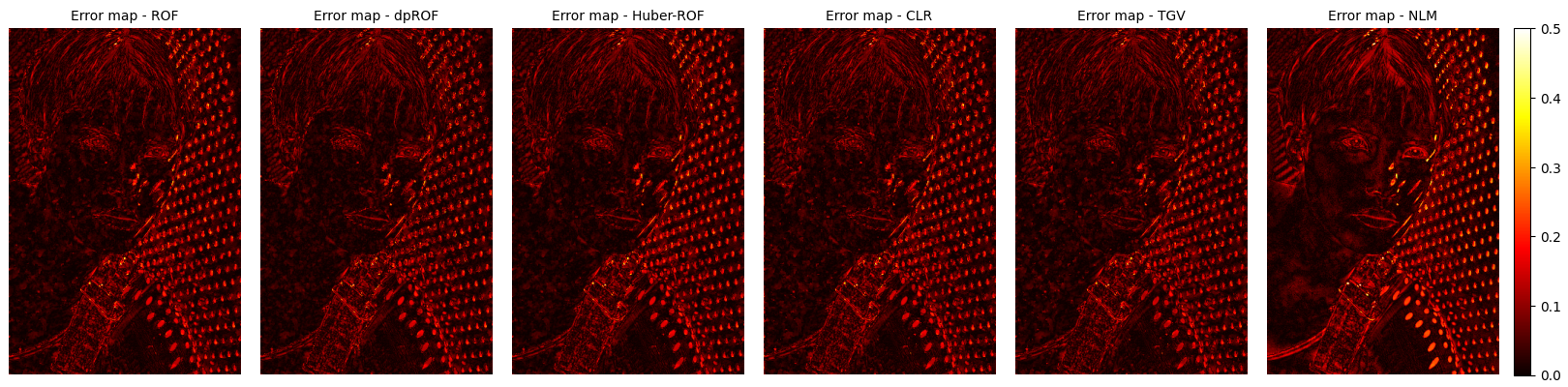}
    \vspace{-8pt} 
    \caption{\small Error maps for the 'lady' image denoising example.}
    \label{fig:lady_error}

    \vspace{10pt} 
    \includegraphics[width=0.9\linewidth]{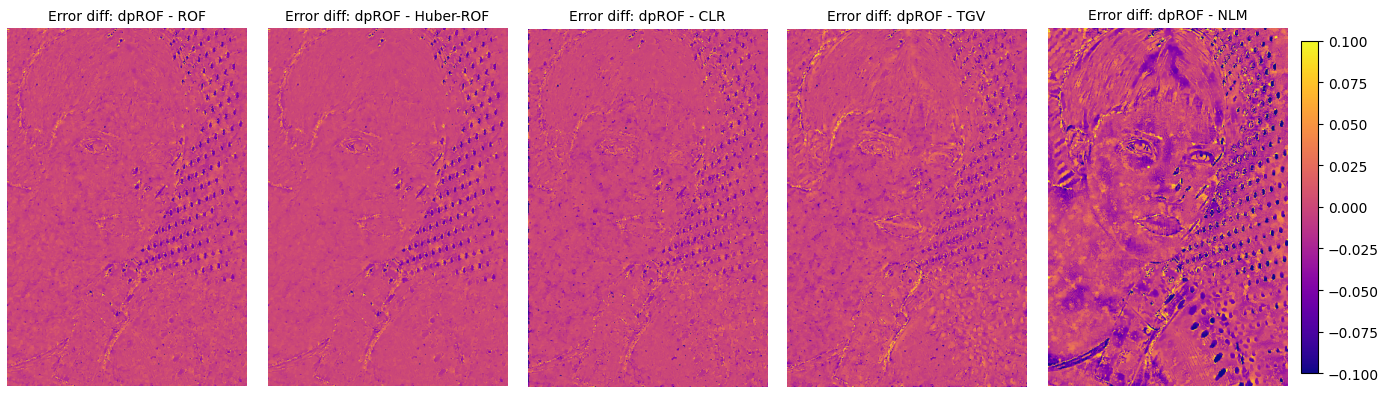}
    \vspace{-8pt} 
    \caption{\small Difference in error between dpROF and other models (blue -- dpROF performs better, yellow -- other model performs better). }
    \label{fig:lady_diff}
\end{figure}

Total Generalized Variation (TGV) is a second-order, non-local variational model given by 
\begin{equation}\label{tgv}  \min_{u}\  TGV_\alpha(u) + \frac{1}{2}\int|u-g|^2 \,\dd x,   
\end{equation} 
\begin{equation} \label{tgv2}
TGV_\alpha(u) = \min_{w} \alpha_1 \int |\nabla u - w|\, \dd x + \frac{\alpha_0}{2} \int |\nabla w + \nabla w^T| \, \dd x,  
\end{equation} 
where $\alpha = (\alpha_1, \alpha_0)$ is a pair of positive constants. Owing to infimal convolution structure of TGV, it is amenable to the robust primal-dual Chambolle--Pock algorithm. We use here an implementation from \cite{recon}.   

NLM is a class of non-variational, non-local denoising models. Here we consider a variant which produces a denoised image by the explicit formula 
\[ u(x) = \frac{1}{Z(x)} \int_{S_a(x)}\!\!\! \exp(-(d(x,y)^2 - 2 \sigma^2)_+/h^2 \sigma^2) g(y)\, \dd y,  \] 
where 
\[ d(x,y) = \frac{1}{b^2}\int_{S_b(0)} |g(x+z) - g(y+z)|^2\, \dd z,\]
$S_a(x)$ is the square of side length $a$ centered at $x$, $a>0$, $b>0$, $h>0$, $\sigma>0$ are parameters, and 
\[Z(x) = \int_{S_a(x)} \exp(-(d(x,y)^2 - 2 \sigma^2)_+/h^2\sigma^2)\, \dd y \]
is a normalization. It produces good results and can be computed very quickly due to algebraic simplification \cite{DCCOJ, Froment}. It corresponds to the 'fast' setting of {\tt denoise\_nl\_means} from the scikit-image library. 

\begin{figure}[h]
    \centering

    \begin{tikzpicture}[font=\scriptsize]
    \begin{axis}[
        width=0.9\textwidth,
        height=0.6\textwidth,
        xmode=log, 
        ymode=linear, 
        xlabel={Noise variance $\sigma^2$},
        ylabel={SSIM},
        xmin=0.0008, xmax=0.083,
        ymin=0.48, ymax=0.94,
        xtick={0.001, 0.002, 0.004, 0.007, 0.01, 0.02, 0.04, 0.07},
        xticklabels={$0.001$, $0.002$, $0.004$, $0.007$, $0.01$, $0.02$, $0.04$, $0.07$},
        ytick={0.6, 0.7, 0.8, 0.9},
        yticklabels={$0.6$, $0.7$, $0.8$, $0.9$},
        grid=major, 
        legend pos=south west, 
        legend style={
            at={(0.02, 0.02)}, 
            anchor=south west, 
            draw=none, 
        },
        mark options={scale=1.2},
        every axis plot/.append style={thick}
    ]

    \addplot[
        color=red,
        mark=square*
    ] coordinates {
        (0.001, 0.9205)
        (0.002, 0.8872)
        (0.004, 0.8415)
        (0.007, 0.8013) 
        (0.01, 0.7721)
        (0.02, 0.7022)
        (0.04, 0.6410)
        (0.07, 0.5891)
    };
    \addlegendentry{dpROF}

    \addplot[
        color=blue,
        mark=triangle*
    ] coordinates {
        (0.001, 0.9175)
        (0.002, 0.8823)
        (0.004, 0.8366)
        (0.007, 0.7915)
        (0.01, 0.7570)
        (0.02, 0.6959)
        (0.04, 0.6303)
        (0.07, 0.5780)
    };
    \addlegendentry{ROF}

    \addplot[
        color=violet, 
        mark=pentagon*
    ] coordinates {
        (0.001, 0.9184)
        (0.002, 0.8868)
        (0.004, 0.8364)
        (0.007, 0.7820)
        (0.01, 0.7651)
        (0.02, 0.6924)
        (0.04, 0.6338)
        (0.07, 0.5788)
    };
    \addlegendentry{TGV}

    \addplot[
        color=green!60!black, 
        mark=diamond*
    ] coordinates {
        (0.001, 0.9034)
        (0.002, 0.8592)
        (0.004, 0.8407)
        (0.007, 0.7929)
        (0.01, 0.7527)
        (0.02, 0.6784)
        (0.04, 0.5899)
        (0.07, 0.5225)
    };
    \addlegendentry{NLM}


    \addplot[
        color=brown, 
        mark=oplus*
    ] coordinates {
        (0.001, 0.9156)
        (0.002, 0.8727)
        (0.004, 0.8357)
        (0.007, 0.7850)
        (0.01, 0.7584)
        (0.02, 0.6889)
        (0.04, 0.6287)
        (0.07, 0.5605)
    };
    \addlegendentry{CLR}

    \addplot[
        color=pink, 
        mark=oplus
    ] coordinates {
        (0.001, 0.9180)
        (0.002, 0.8787)
        (0.004, 0.8345)
        (0.007, 0.7930)
        (0.01, 0.7627)
        (0.02, 0.6990)
        (0.04, 0.6338)
        (0.07, 0.5809)
    };
   \addlegendentry{Huber-ROF}

    \end{axis}
    \end{tikzpicture}
    \vspace{-8pt} 
    \caption{\small Comparison of ROF, dpROF, TGV, NLM, CLR and Huber-ROF models across different noise levels (SSIM)}
    \label{fig:ssim_comparison_graph}
\end{figure} 

\begin{figure}[h] 
\centering
     
    \begin{tikzpicture}[font=\scriptsize]
    \begin{axis}[
        width=0.9\textwidth,
        height=0.6\textwidth,
        xmode=log, 
        ymode=linear, 
        xlabel={Noise variance $\sigma^2$},
        ylabel={LPIPS},
        xmin=0.0008, xmax=0.083,
        ymin=0.0, ymax=0.55,
        xtick={0.001, 0.002, 0.004, 0.007, 0.01, 0.02, 0.04, 0.07},
        xticklabels={$0.001$, $0.002$, $0.004$, $0.007$, $0.01$, $0.02$, $0.04$, $0.07$},
        ytick={0.1, 0.2, 0.3, 0.4, 0.5},
        yticklabels={$0.1$, $0.2$, $0.3$, $0.4$, $0.5$},
        grid=major, 
        legend pos=north west, 
        legend style={
            at={(0.02, 0.98)}, 
            anchor=north west, 
            draw=none, 
        },
        mark options={scale=1.2},
        every axis plot/.append style={thick}
    ]
    
    \addplot[
        color=red,
        mark=square*
    ] coordinates {
        (0.001, 0.0824)
        (0.002, 0.1238)
        (0.004, 0.1808)
        (0.007, 0.2262) 
        (0.01, 0.2615)
        (0.02, 0.3278)
        (0.04, 0.3911)
        (0.07, 0.4442)
    };
    \addlegendentry{dpROF}

    \addplot[
        color=blue,
        mark=triangle*
    ] coordinates {
        (0.001, 0.0832)
        (0.002, 0.1250)
        (0.004, 0.1777)
        (0.007, 0.2265)
        (0.01, 0.2604)
        (0.02, 0.3289)
        (0.04, 0.3986)
        (0.07, 0.4550)
    };
    \addlegendentry{ROF}

    \addplot[
        color=violet, 
        mark=pentagon*
    ] coordinates {
        (0.001, 0.0823)
        (0.002, 0.1131)
        (0.004, 0.1678)
        (0.007, 0.2335)
        (0.01, 0.2559)
        (0.02, 0.3377)
        (0.04, 0.4029)
        (0.07, 0.4572)
    };
    \addlegendentry{TGV}
    
    \addplot[
        color=green!60!black, 
        mark=diamond*
    ] coordinates {
        (0.001, 0.0898)
        (0.002, 0.1367)
        (0.004, 0.1385)
        (0.007, 0.1874)
        (0.01, 0.2262)
        (0.02, 0.3122)
        (0.04, 0.4283)
        (0.07, 0.5171)
    };
    \addlegendentry{NLM}


    \addplot[
        color=brown, 
        mark=oplus*
    ] coordinates {
        (0.001, 0.0842)
        (0.002, 0.1318)
        (0.004, 0.1763)
        (0.007, 0.2305)
        (0.01, 0.2601)
        (0.02, 0.3334)
        (0.04, 0.4033)
        (0.07, 0.4693)
    };
    \addlegendentry{CLR}

    \addplot[
        color=pink, 
        mark=oplus
    ] coordinates {
        (0.001, 0.0841)
        (0.002, 0.1267)
        (0.004, 0.1776)
        (0.007, 0.2289)
        (0.01, 0.2591)
        (0.02, 0.3276)
        (0.04, 0.3968)
        (0.07, 0.4526)
    };
    \addlegendentry{Huber-ROF}

    \end{axis}
    \end{tikzpicture}
    \vspace{-8pt} 
    \caption{\small Comparison of ROF, dpROF, TGV, NLM, CLR and Huber-ROF models across different noise levels (LPIPS)}
    \label{fig:lpips_comparison_graph}
\end{figure}

We compare the performance of optimized models by a single run on the 'test' subset of BSDS500, in terms of average values of SSIM, LPIPS, PSNR as well as execution time, with stopping condition that the difference between successive iterations is smaller than a given threshold $\varepsilon$. In the case of dpROF with initial ROF step, we report separately elapsed time for the two steps. In order to speed up the model, we choose threshold $\sqrt{\varepsilon}$ for the initial step (this has negligible effect on performance, see \cite{github}). We report the results for $\sigma^2 = 0.002, \ 0.01,\ 0.04$ in Tables \ref{tab:comparison0002,1e-4}-\ref{tab:comparison004,1e-4}. In addition to the simplified dpROF model with values of parameters determined by Table \ref{tab:dprof_par}, we include several other variants of dpROF. Here we summarize the included models: 
\begin{itemize} 
\item dpROF: Double-phase Rudin--Osher--Fatemi model with parameters given by Table \ref{tab:dprof_par}.

\item ROF: The classical Rudin--Osher--Fatemi model, optimized with respect to the parameter $\lambda$.

\item TGV: The total generalized variation model, optimized w.r.\ to parameters $\alpha_1, \alpha_2$.

\item NLM: The non-local means model, optimized w.r.\ to the parameter $h$. Patch size and range are chosen as in \cite{BCM2011}. 

\item CLR: The Chen--Levine--Rao model, optimized w.r.\ to $\lambda$, with $K$ and $\beta$ chosen as in \cite{CLR}. 

\item Huber-ROF: The Huber-ROF model, optimized w.r.\ to $\lambda$ and $\alpha$.

\item dpROF-noisy: The double-phase ROF model with the weight computed directly from the noisy datum. More precisely, in the definition of the weight in the dpROF model, instead of setting $w(x) = W(|\nabla \tilde{u}_{\rm ROF}|(x))$, we simply use $w(x) = W(|\nabla \tilde{g}|(x))$, i.e., the mollified noisy datum. We then use the same parameters as in the main version of dpROF.

\item dpROF-SSIM: The dpROF model, but instead of parameters $\lambda_0$, $\lambda$, $a$ and $b/a$ given by Table \ref{tab:dprof_par}, we use values corresponding to maximal average SSIM on BSDS500 for the given noise value (see Sections 2.1 and 2.2 in the 'comparison' supplement in \cite{github}).

\item dpROF-LPIPS: The dpROF model, but instead of parameters $\lambda_{0}$, $\lambda$, $a$ and $b/a$ given by Table \ref{tab:dprof_par}, we use values corresponding to maximal average LPIPS on BSDS500 for the given noise value (see Sections 2.1 and 2.2 in the 'comparison' supplement in \cite{github}).

\item dpROF-edge: The double-phase ROF model with the weight computed from an edge map obtained from the noisy datum using a standard edge-detection algorithm. In the results in the table we used the Canny edge detection \cite{Canny} (in the initial tests it proved superior to Sobel/Scharr edge detection for this purpose). We use the same parameters as in the main version of dpROF.
\end{itemize} 

We plot the values of SSIM and LPIPS for considered models in Figures \ref{fig:ssim_comparison_graph} and \ref{fig:lpips_comparison_graph}. Double-phase ROF consistently outperforms other methods in terms of SSIM and PSNR. The values of LPIPS are comparable to other models, with dpROF taking clear lead for high levels of noise ($\sigma^2 = 0.04$ and larger). For low $\sigma$, the best results in terms of LPIPS are produced by TGV, while for medium $\sigma$ -- by NLM. 

Finally, we include for illustration (Figure \ref{fig:lady_denoise}) comparison of denoising results on one of images ('lady') from the BSDS500 dataset together with error maps (Figure \ref{fig:lady_error}). We also include Figure \ref{fig:lady_diff}, where the difference of errors between dpROF and other models is presented. Especially the latter explains qualitative differences between the models. The blue splotches across the areas of the lady's skin confirm that dpROF reduces staircasing better than other tested models. Incidentally, it also turns out to significantly outperform other models at restoring certain kinds of small-size structures. On the other hand, we observe that it underperforms in terms of recovering object contours.

\begin{table}[H] 
    \centering
    \caption{Comparison of denoising methods: $\sigma^2 = 0.002$, $\varepsilon = 10^{-4}$}
    \label{tab:comparison0002,1e-4}
    \begin{tabular}{lcccc}
        \toprule
        & \multicolumn{4}{c}{Average values of metrics} \\
        \cmidrule(lr){2-5}
        Method & Time & SSIM $\uparrow$ & LPIPS $\downarrow$ & PSNR $\uparrow$ \\
        \midrule
        dpROF           & 0.06s + 0.99s & 0.8871 & 0.1240 & 31.45 \\
        ROF             & 0.27s        & 0.8823 & 0.1250 & 31.26 \\
        TGV             & 2.37s        & 0.8867 & 0.1131 & 31.37 \\
        NLM        & 1.07s        & 0.8592 & 0.1367 & 30.38 \\ 
        CLR & 1.24s        & 0.8789 & 0.1249 & 31.18 \\
        Huber-ROF       & 0.79s        & 0.8788 & 0.1266 & 31.00 \\
        dpROF-noisy     & 1.17s        & 0.8829 & 0.1262 & 31.32 \\
        dpROF-SSIM      & 0.06s + 1.09s & 0.8853 & 0.1249 & 31.42 \\
        dpROF-LPIPS     & 0.09s + 1.09s & 0.8861 & 0.1236 & 31.40 \\
        dpROF-edge      & 2.22s        & 0.8721 & 0.1431 & 30.33 \\
        \bottomrule
    \end{tabular}
\end{table}

\begin{table}[H] 
    \centering
    \vspace{-8pt}
    \caption{Comparison of denoising methods: $\sigma^2 = 0.01$, $\varepsilon = 10^{-4}$}
    \label{tab:comparison001,1e-4}
    \begin{tabular}{lcccc}
        \toprule
        & \multicolumn{4}{c}{Average values of metrics} \\
        \cmidrule(lr){2-5}
        Method & Time & SSIM $\uparrow$ & LPIPS $\downarrow$ & PSNR $\uparrow$ \\
        \midrule
        dpROF           & 0.07s + 1.26s & 0.7720 & 0.2618 & 27.38 \\
        ROF             & 0.42s        & 0.7619 & 0.2615 & 27.23 \\
        TGV             & 6.74s        & 0.7651 & 0.2559 & 27.27 \\
        NLM        & 0.99s        & 0.7527 & 0.2262 & 27.21 \\
        CLR & 8.25s        & 0.7584 & 0.2600 & 27.11 \\
        Huber-ROF       & 1.22s        & 0.7628 & 0.2591 & 27.19 \\
        dpROF-noisy     & 1.50s        & 0.7472 & 0.2758 & 27.20 \\
        dpROF-SSIM      & 0.12s + 1.43s & 0.7639 & 0.2605 & 27.25 \\
        dpROF-LPIPS     & 0.13s + 1.79s & 0.7698 & 0.2569 & 27.33 \\
        dpROF-edge      & 3.14s        & 0.7404 & 0.2960 & 26.75 \\
        \bottomrule
    \end{tabular}
\end{table}

\begin{table}[H] 
    \centering
    \vspace{-8pt} 
    \caption{Comparison of denoising methods: $\sigma^2 = 0.04$, $\varepsilon = 10^{-4}$}
    \label{tab:comparison004,1e-4}
    \begin{tabular}{lcccc}
        \toprule
        & \multicolumn{4}{c}{Average values of metrics} \\
        \cmidrule(lr){2-5}
        Method & Time & SSIM $\uparrow$ & LPIPS $\downarrow$ & PSNR $\uparrow$ \\
        \midrule
        dpROF           & 0.09s + 1.67s  & 0.6412 & 0.3910 & 24.20 \\
        ROF            & 0.64s        & 0.6309 & 0.3993 & 24.00 \\
        TGV             & 6.95s        & 0.6337 & 0.4031 & 24.13 \\
        NLM        & 2.80s       & 0.5899  & 0.4283 & 23.36 \\
        CLR & 4.32s        & 0.5810 & 0.5213 & 22.73 \\
        Huber-ROF       & 1.72s        & 0.6339 & 0.3968 & 24.00 \\
        dpROF-noisy     & 1.96s        & 0.5827 & 0.4340 & 23.70 \\
        dpROF-SSIM      & 0.09s + 2.02s & 0.6432 & 0.4057 & 24.12 \\
        dpROF-LPIPS     & 0.17s + 2.37s & 0.6443 & 0.3958 & 24.18 \\
        dpROF-edge      & 4.31s        & 0.5900 & 0.4454 & 23.67 \\
        \bottomrule
    \end{tabular}
\end{table}

\section{Multichannel images} \label{sec:multi} 

We consider three different extensions of the dpROF model to the color image case: 

\begin{itemize} 
\item RGB-dpROF: Each RGB channel is processed separately by the dpROF model. In particular, every channel is processed using a different adaptive weight. 
\item YUV-dpROF: We work in the YUV color scheme. For each channel, we solve the dpROF minimization problem \eqref{eq:dprof} with weight $w(x)$ computed by formula \eqref{eq:weight} from the Y (luminance) channel of the datum. 
\item Vectorial dpROF: we proceed as in the single-channel dpROF model with the Euclidean norm $|\nabla u|$ replaced by the Frobenius norm $\|\nabla u\|_F$ given by
\[ \|\nabla u\|^2_F = \textstyle\sum_{i=1}^3 |\nabla u^i|^2 \]
at all steps (in the fidelity term, $|u-g|$ is now the Euclidean norm). See \cite{github} for details. 
\end{itemize} 
 
\begin{table}[h] 
    \centering
    \renewcommand{\arraystretch}{1.5}
    \caption{Choice of parameters in the vectorial dpROF model}
    \label{tab:vec_par}
    \begin{tabular}{c|cccc}
        \hline
        Noise level & $\lambda_{\rm ROF}$ & $\lambda$ & $a$ & $b/a$\\
        \hline
        Small ($\sigma^2 \leq 0.005$) & \multirow{2}{*}{$1.0 \cdot \sigma$} & $0.8 \cdot \sigma$ & \multirow{3}{*}{100} & 85 \\
        Medium ($0.005 < \sigma^2 \leq 0.015$) & & \multirow{2}{*}{$0.9 \cdot \sigma$} & & 50 \\
        Large ($\sigma^2 > 0.015$) &  $1.2 \cdot \sigma$ &  & & 50 \\
        \hline
    \end{tabular}
\end{table}

\begin{table}[H] 
    \centering
    \caption{Multichannel denoising comparison: $\sigma^2 = 0.002$}
    \label{tab:vectorial0002}
    \begin{tabular}{lcccc}
        \toprule
        & \multicolumn{4}{c}{Metrics} \\
        \cmidrule(lr){2-5}
        Method & Time & SSIM $\uparrow$ & LPIPS $\downarrow$ & PSNR $\uparrow$ \\
        \midrule
        ROF          & 0.83s      & 0.8970 & 0.0876 & 31.99 \\
        dpROF        & 0.27s + 1.55s      & 0.9008 & 0.0796 & 32.15 \\
        NL-means     & 0.26s      & 0.8928 & 0.0916 & 31.21 \\        
        YUV-dpROF    & 4.77s      & 0.8986 & 0.1012 & 32.22 \\
        RGB-dpROF    & 4.68s      & 0.8839 & 0.0937 & 31.45 \\
        \bottomrule
    \end{tabular}
\end{table}

\begin{table}[H] 
    \centering
    \caption{Multichannel denoising comparison: $\sigma^2 = 0.007$}
    \label{tab:vectorial0007}
    \begin{tabular}{lcccc}
        \toprule
        & \multicolumn{4}{c}{Metrics} \\
        \cmidrule(lr){2-5}
        Method & Time & SSIM $\uparrow$ & LPIPS $\downarrow$ & PSNR $\uparrow$ \\
        \midrule
        ROF          & 1.01s      & 0.8114 & 0.1768 & 28.73 \\
        dpROF        & 0.32s + 1.92s    & 0.8215 & 0.1744 & 28.84 \\
        NL-means     & 0.29s      & 0.8101 & 0.1635 & 27.96 \\
        YUV-dpROF    & 6.35s      & 0.8089 & 0.2017 & 28.74 \\
        RGB-dpROF    & 5.87s      & 0.7935 & 0.1983 & 28.16 \\
        \bottomrule
    \end{tabular}
\end{table}

\begin{table}[H] 
    \centering
    \caption{Multichannel denoising comparison: $\sigma^2 = 0.04$}
    \label{tab:vectorial004}
    \begin{tabular}{lcccc}
        \toprule
        & \multicolumn{4}{c}{Metrics} \\
        \cmidrule(lr){2-5}
        Method & Time & SSIM $\uparrow$ & LPIPS $\downarrow$ & PSNR $\uparrow$ \\
        \midrule
        ROF          & 1.99s      & 0.6529 & 0.3673 & 24.28 \\
        dpROF        & 0.51s + 2.71s      & 0.6639 & 0.3647 & 24.44 \\
        NL-means     & 0.33s      & 0.6274 & 0.3406 & 24.17 \\
        YUV-dpROF    & 9.72s      & 0.6372 & 0.3929 & 24.09 \\
        RGB-dpROF    & 9.11s      & 0.6285 & 0.3978 & 23.98 \\
        \bottomrule
    \end{tabular}
\end{table}

In the first two extensions, we use the same parameters as in the single-channel case. For vectorial dpROF, we select the parameters based on optimization over the 'train' and 'val' subset of BSDS500 as in Section \ref{sec:comparison}. We fix the values of parameters as a function of the noise variance given in Table \ref{tab:vec_par}. The detailed results of optimization are reported in supplement 'RGB' in \cite{github}. We include here Table \ref{tab:sens004} and Figure \ref{fig:heat004} to illustrate the sensitivity of the model with respect to the change of parameters. 

We compare the performance of the models by a single run on the 'test' subset of BSDS500 with added Gaussian noise, in terms of average values of SSIM, LPIPS and PSNR distance from the clean image, as well as execution time, similarly as in section \ref{sec:comparison}. We present the results for selected noise values in Tables \ref{tab:vectorial0002}, \ref{tab:vectorial0007}, \ref{tab:vectorial004}, and over the whole tested noise range in Figures \ref{fig:graphs_color_ssim} and \ref{fig:graphs_color_lpips}. 
As in the grayscale case, vectorial dpROF consistently outperforms other tested methods in terms of SSIM. In terms of LPIPS, it leads in the low noise range, but it is topped by NLM for higher noise levels. 

Vectorial dpROF consistently beats the other tested multichannel versions of dpROF in terms of similarity metrics as well as execution time. We will therefore use vectorial dpROF as \emph{the} generalization of dpROF to multichannel case, and refer to it simply as 'dpROF'.

\begin{table}[H]
\centering
\caption{Sensitivity of vectorial dpROF with respect to change of $\lambda$ for noise variance $\sigma^2=0.04$.}
\label{tab:sens004}
\begin{tabular}{cccc}
\toprule
$\lambda/\sigma$ & LPIPS $\downarrow$ & SSIM $\uparrow$ & PSNR $\uparrow$ \\
\midrule
0.7 & 0.370854 & 0.658509 & 24.5933 \\
0.8 & 0.364464 & 0.665332 & 24.6558 \\
0.9 & \textbf{0.361473} & \textbf{0.667786} & 24.6888 \\
1.0 & 0.361560 & 0.666478 & \textbf{24.6911} \\
1.1 & 0.364290 & 0.662201 & 24.6290 \\
1.2 & 0.367846 & 0.656302 & 24.5316 \\
\bottomrule
\end{tabular}
\end{table}

\begin{figure}[H]
\centering

\subfloat[LPIPS]{
\centering
\includegraphics[width=0.32\linewidth]
{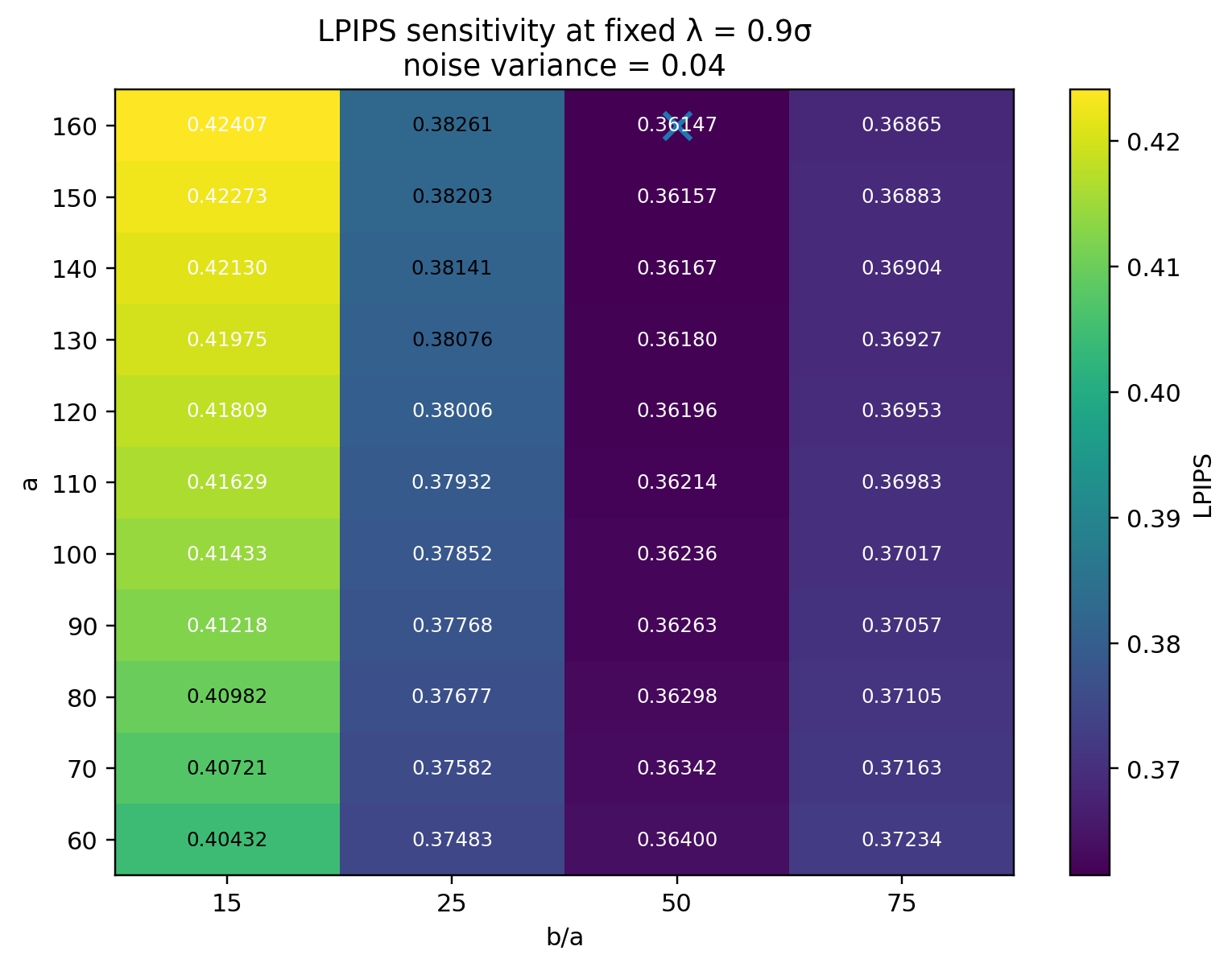}}
\subfloat[SSIM]{
\centering
\includegraphics[width=0.32\linewidth]
{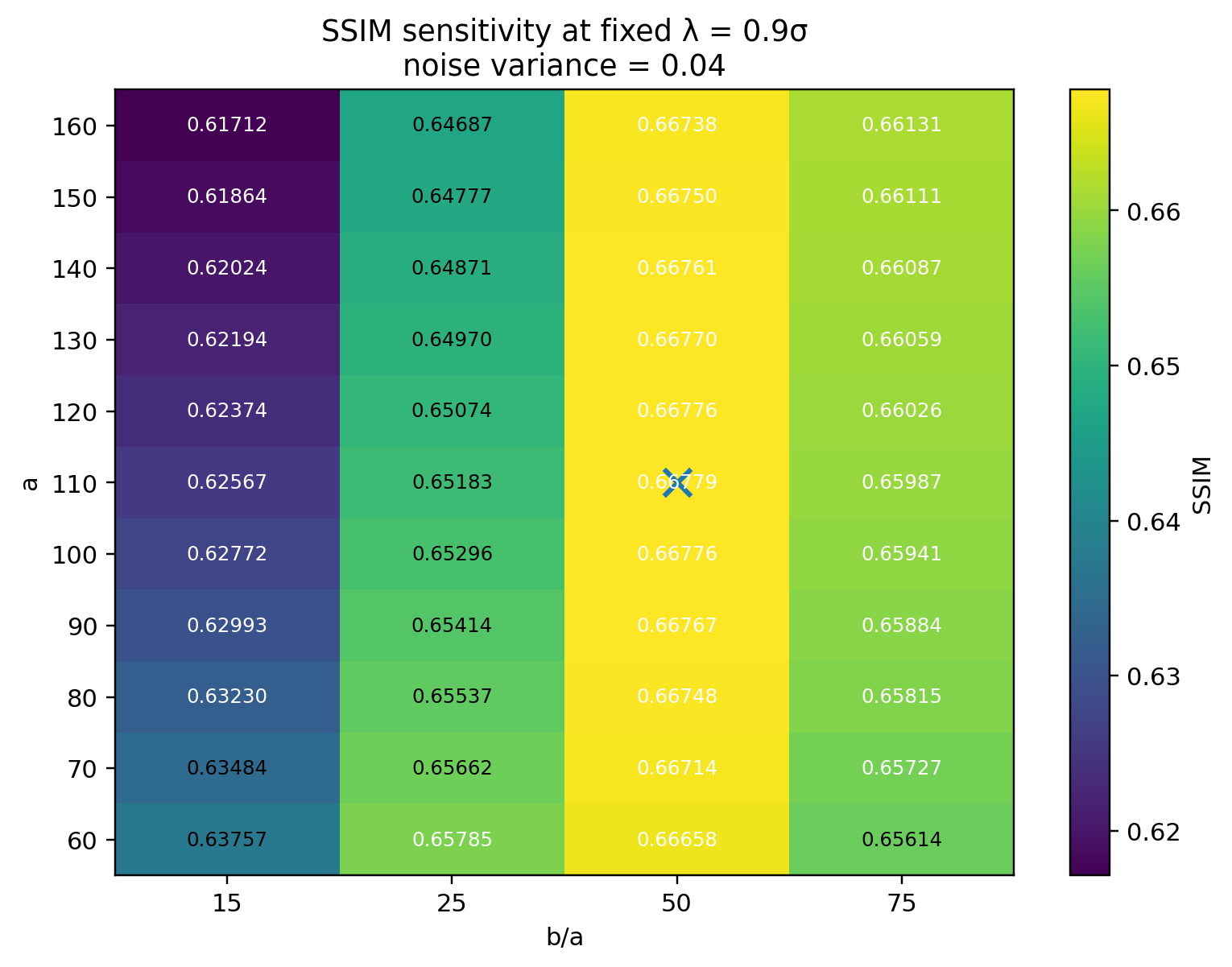}}
\subfloat[PSNR]{
\centering
\includegraphics[width=0.32\linewidth]
{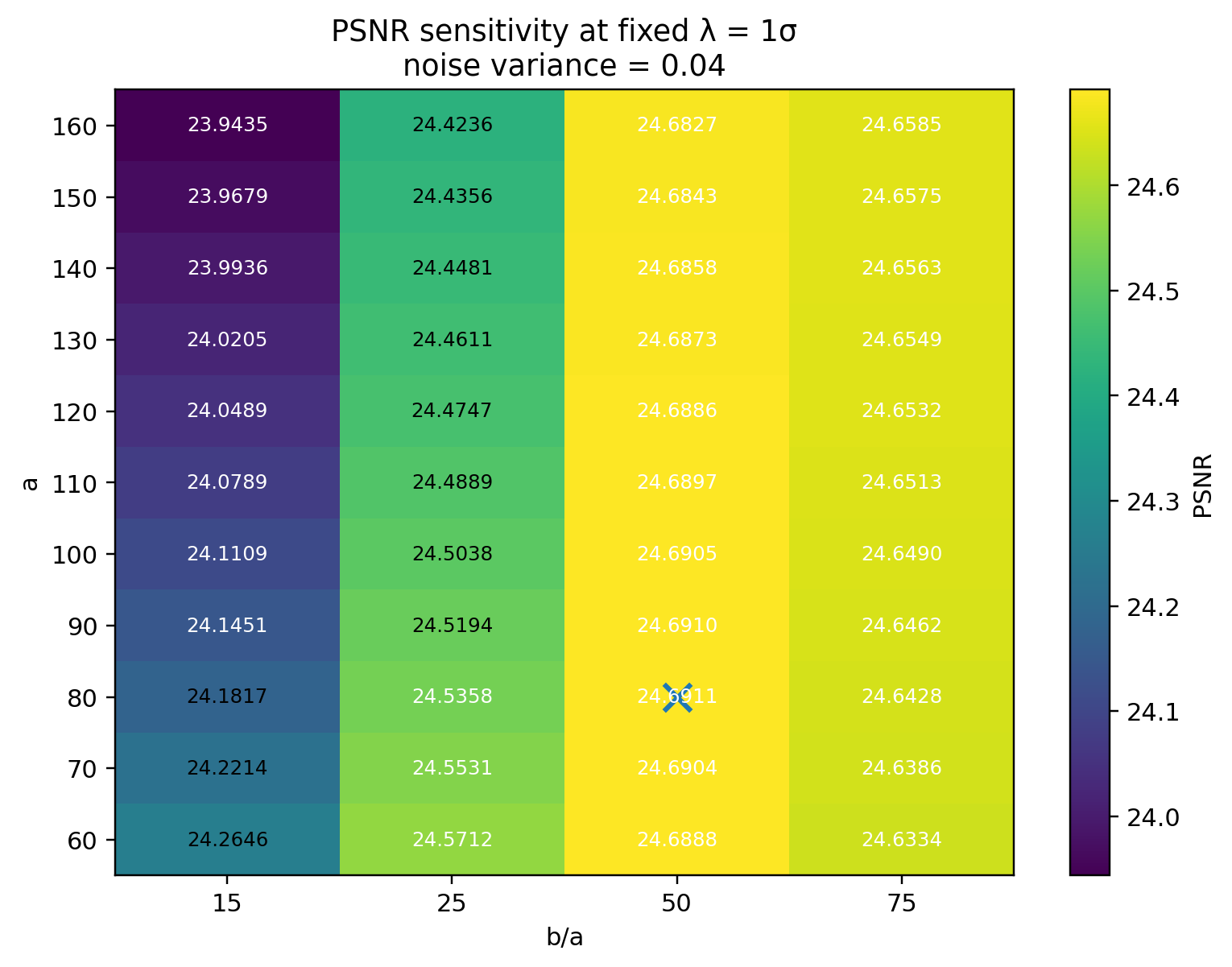}}

\caption{Sensitivity of vectorial dpROF with respect to $a$ and $b/a$ for noise variance $\sigma^2=0.04$.}
\label{fig:heat004}
\end{figure}

\begin{figure}[H]
    \centering
    
   \begin{tikzpicture}[font=\scriptsize]
    \begin{axis}[
        width=0.9\textwidth,
        height=0.6\textwidth,
        xmode=log,
        ymode=linear,
        xlabel={Noise variance $\sigma^2$},
        ylabel={SSIM},
        xmin=0.0008, xmax=0.083,
        ymin=0.5, ymax=0.95,
        xtick={0.001, 0.002, 0.004, 0.007, 0.01, 0.02, 0.04, 0.07},
        xticklabels={$0.001$, $0.002$, $0.004$, $0.007$, $0.01$, $0.02$, $0.04$, $0.07$},
        ytick={0.6, 0.7, 0.8, 0.9},
        yticklabels={$0.6$, $0.7$, $0.8$, $0.9$},
        grid=major,
        legend pos=south west,
        legend style={
            at={(0.02, 0.02)},
            anchor=south west,
            draw=none,
        },
        mark options={scale=1.2},
        every axis plot/.append style={thick}
    ]
    
    \addplot[
        color=red,
        mark=square*
    ] coordinates {
        (0.001, 0.9312)
        (0.002, 0.9008)
        (0.004, 0.8568)
        (0.007, 0.8215) 
        (0.01, 0.7934)
        (0.02, 0.7315)
        (0.04, 0.6639)
        (0.07, 0.6079)
    };
    \addlegendentry{dpROF}

    \addplot[
        color=blue,
        mark=triangle*
    ] coordinates {
        (0.001, 0.9289)
        (0.002, 0.8970)
        (0.004, 0.8545)
        (0.007, 0.8114)
        (0.01, 0.7797)
        (0.02, 0.7203)
        (0.04, 0.6529)
        (0.07, 0.5962)
    };
    \addlegendentry{ROF}

    \addplot[
        color=green!60!black,
        mark=diamond*
    ] coordinates {
        (0.001, 0.9236)
        (0.002, 0.8928)
        (0.004, 0.8520)
        (0.007, 0.8101)
        (0.01, 0.7788)
        (0.02, 0.7084)
        (0.04, 0.6274)
        (0.07, 0.5577)
    };
    \addlegendentry{NLM}

    \addplot[
        color=cyan, 
        mark=star
    ] coordinates {
        (0.001, 0.9283)
        (0.002, 0.8986)
        (0.004, 0.8603)
        (0.007, 0.8089)
        (0.01, 0.7822)
        (0.02, 0.7001)
        (0.04, 0.6372)
        (0.07, 0.5824)
    };
    \addlegendentry{YUV-dpROF}

    \addplot[
        color=magenta, 
        mark=asterisk
    ] coordinates {
        (0.001, 0.9187)
        (0.002, 0.8839)
        (0.004, 0.8370)
        (0.007, 0.7935)
        (0.01, 0.7635)
        (0.02, 0.6908)
        (0.04, 0.6285)
        (0.07, 0.5760)
    };
    \addlegendentry{RGB-dpROF}

    \end{axis}
    \end{tikzpicture} 
    \vspace{-8pt} 
    \caption{\small Comparison of vectorial ROF, dpROF, NLM, YUV- and RGB-dpROF models across different noise levels (SSIM)}
    \label{fig:graphs_color_ssim}

    \end{figure} 

    \begin{figure}[H]
    \centering
    \begin{tikzpicture}[font=\scriptsize] 
    \begin{axis}[
        width=0.9\textwidth,
        height=0.6\textwidth,
        xmode=log,
        ymode=linear,
        xlabel={Noise variance $\sigma^2$},
        ylabel={LPIPS},
        xmin=0.0008, xmax=0.083,
        ymin=0.0, ymax=0.5,
        xtick={0.001, 0.002, 0.004, 0.007, 0.01, 0.02, 0.04, 0.07},
        xticklabels={$0.001$, $0.002$, $0.004$, $0.007$, $0.01$, $0.02$, $0.04$, $0.07$},
        ytick={0.1, 0.2, 0.3, 0.4, 0.5},
        yticklabels={$0.1$, $0.2$, $0.3$, $0.4$, $0.5$},
        grid=major,
        legend pos=north west,
        legend style={
            at={(0.02, 0.98)},
            anchor=north west,
            draw=none,
        },
        mark options={scale=1.2},
        every axis plot/.append style={thick}
    ]
    
    \addplot[
        color=red,
        mark=square*
    ] coordinates {
        (0.001, 0.0543)
        (0.002, 0.0796)
        (0.004, 0.1217)
        (0.007, 0.1744) 
        (0.01, 0.2057)
        (0.02, 0.2829)
        (0.04, 0.3647)
        (0.07, 0.4311)
    };
    \addlegendentry{dpROF}

    \addplot[
        color=blue,
        mark=triangle*
    ] coordinates {
        (0.001, 0.0586)
        (0.002, 0.0876)
        (0.004, 0.1301)
        (0.007, 0.1768)
        (0.01, 0.2124)
        (0.02, 0.2895)
        (0.04, 0.3673)
        (0.07, 0.4331)
    };
    \addlegendentry{ROF}

    \addplot[
        color=green!60!black,
        mark=diamond*
    ] coordinates {
        (0.001, 0.0653)
        (0.002, 0.0916)
        (0.004, 0.1273)
        (0.007, 0.1635)
        (0.01, 0.1900)
        (0.02, 0.2631)
        (0.04, 0.3406)
        (0.07, 0.4221)
    };
    \addlegendentry{NLM}

    \addplot[
        color=cyan,
        mark=star
    ] coordinates {
        (0.001, 0.0712)
        (0.002, 0.1012)
        (0.004, 0.1413)
        (0.007, 0.2017)
        (0.01, 0.2307)
        (0.02, 0.3231)
        (0.04, 0.3929)
        (0.07, 0.4560)
    };
    \addlegendentry{YUV-dpROF}

    \addplot[
        color=magenta,
        mark=asterisk
    ] coordinates {
        (0.001, 0.0640)
        (0.002, 0.0937)
        (0.004, 0.1389)
        (0.007, 0.1983)
        (0.01, 0.2315)
        (0.02, 0.3267)
        (0.04, 0.3978)
        (0.07, 0.4608)
    };
    \addlegendentry{RGB-dpROF}

    \end{axis}
    \end{tikzpicture} 
    \vspace{-8pt}
    \caption{\small Comparison of vectorial ROF, dpROF, NLM, YUV- and RGB-dpROF models across different noise levels (LPIPS)}
    \label{fig:graphs_color_lpips}
\end{figure}

\begin{figure}[H]
    \centering
    \includegraphics[width=0.9\linewidth]{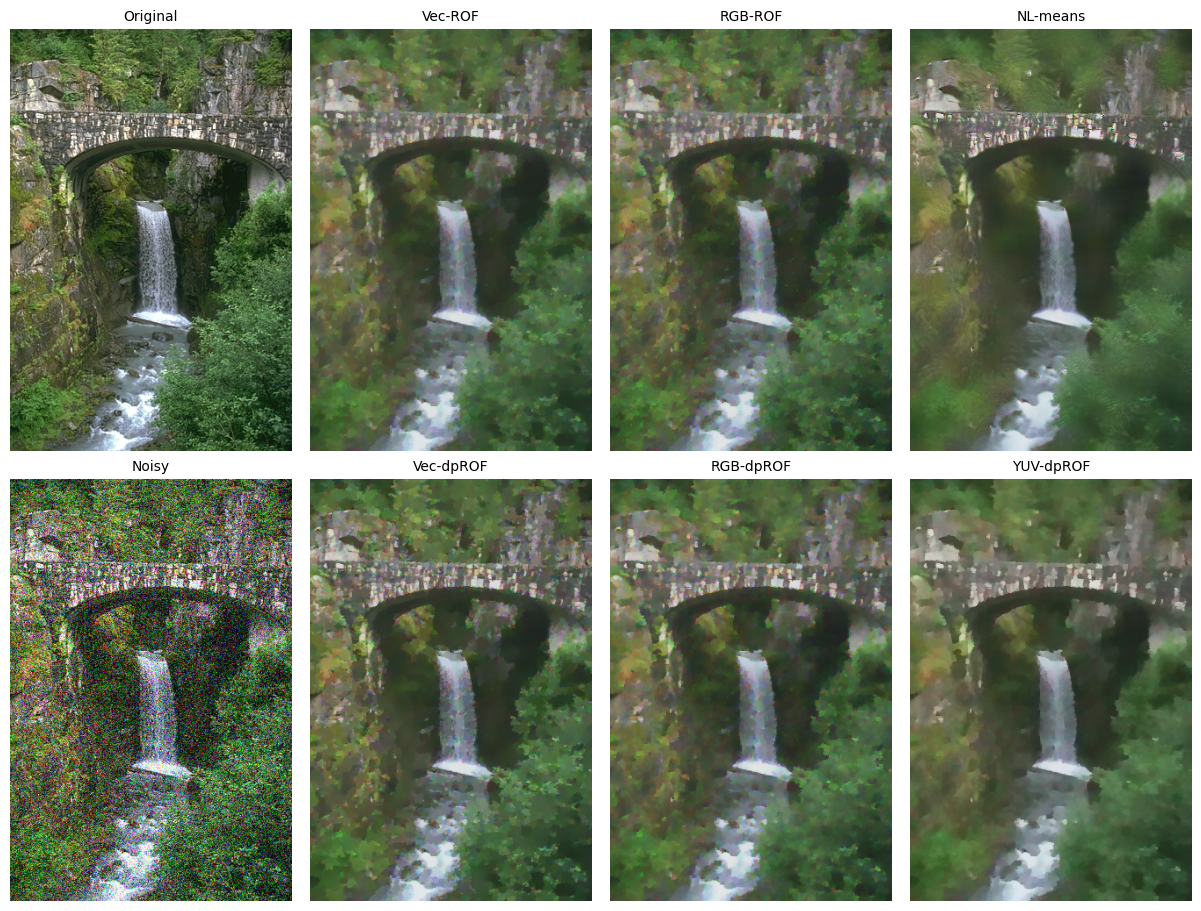}
    \vspace{-8pt} 
    \caption{\small Original 'waterfall' image; image with added Gaussian noise with $\sigma^2=0.04$; optimized denoising results, respectively: Vec-ROF (${\rm SSIM} = 0.475$, ${\rm LPIPS} = 0.543$), RGB-ROF (${\rm SSIM} = 0.449$, ${\rm LPIPS} = 0.561$), NLM (${\rm SSIM} = 0.432$, ${\rm LPIPS} = 0.557$), Vec-dpROF (${\rm SSIM} = 0.509$, ${\rm LPIPS} = 0.506$), YUV-dpROF (${\rm SSIM} = 0.419$, ${\rm LPIPS} = 0.600$), RGB-dpROF (${\rm SSIM} = 0.445$, ${\rm LPIPS} = 0.586$).}
        \label{fig:waterfall_denoise}

\centering
\vspace{10pt} 
    \includegraphics[width=0.9\linewidth]{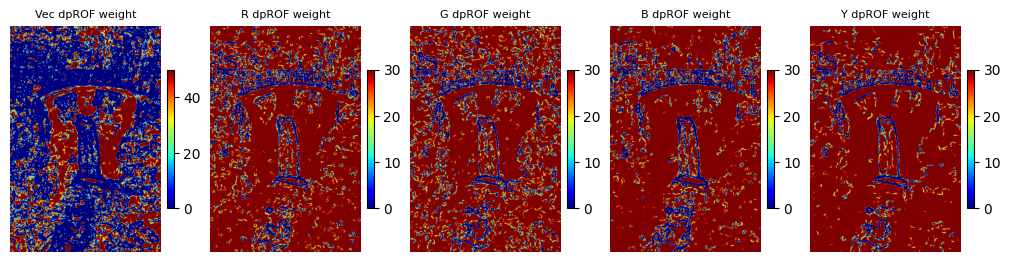}
    \vspace{-8pt} 
    \caption{\small Comparison of the weight $w(x)$ produced by the vectorial ROF minimizer and by the scalar ROF minimizers from different channels.}
    \label{fig:waterfall_weights} 

    \vspace{10pt} 
    \includegraphics[width=0.9\linewidth]{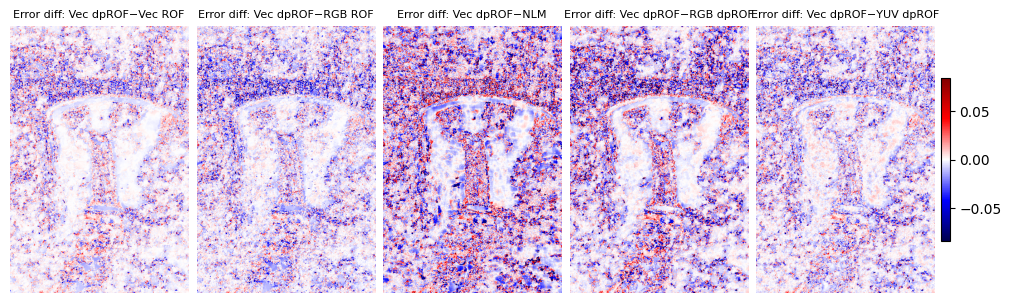}
    \vspace{-8pt} 
    \caption{\small Difference in error between Vec-dpROF and other models (blue -- Vec-dpROF performs better, red -- other model performs better). }
    \label{fig:waterfall_diff}
\end{figure}

We also provide for illustration the comparison of results of tested models applied to a 'waterfall' image from the BSDS500 dataset. We observe that vectorial dpROF fares comparatively well in reproducing irregularly textured regions such as trees. On the other hand, it is well known that NLM excels at reproducing recurring patterns. Conversely, irregular textures may be treated by NLM as noise and suppressed, leading to blurring and artifacts. We believe this could explain why NLM outperforms other models, including dpROF, in terms of LPIPS, even if it has worse SSIM --- as LPIPS is based on deep features, it may favor patterns and ignore (possibly large) regions with irregular textures. This however carries a risk if a desirable feature is not a pattern favored by NLM. Since ROF-type methods are local (as opposed to NL means), they do not carry this type of bias.

\begin{table*}[h] 
\centering
\caption{Comparison of dpROF and ROF on different color datasets using the corresponding 
BSDS500-derived parameter prescriptions without dataset-specific retuning.}
\label{tab:veccross}
\resizebox{\textwidth}{!}{%
\begin{tabular}{lllcccccccc}
\toprule
\multirow{2}{*}{Model}& \multirow{2}{*}{Dataset}
& \multirow{2}{*}{Metric}
& \multicolumn{8}{c}{Noise variance $\sigma^2$} \\
\cmidrule(lr){4-11}
& & & 0.001 & 0.002 & 0.004 & 0.007 & 0.01 & 0.02 & 0.04 & 0.07 \\
\midrule

\multirow{2}{*}{dpROF} & \multirow{2}{*}{Kodak24}
& SSIM $\uparrow$
& \textbf{0.9207} & \textbf{0.8907} & \textbf{0.8472} & \textbf{0.8167} & \textbf{0.7911} & \textbf{0.7385} & \textbf{0.6799} & \textbf{0.6339} \\
& & LPIPS $\downarrow$
& 0.0801 & \textbf{0.1081} & \textbf{0.1540} & \textbf{0.2122} & \textbf{0.2446} & \textbf{0.3207} & \textbf{0.4001} & \textbf{0.4626} \\

\midrule

\multirow{2}{*}{ROF} & \multirow{2}{*}{Kodak24}
& SSIM $\uparrow$
& 0.9186 & 0.8862 & 0.8439 & 0.8012 & 0.7698 & 0.7271 & 0.6685 & 0.6228 \\
& & LPIPS $\downarrow$
& \textbf{0.0799} & 0.1128 & 0.1610 & 0.2137 & 0.2535 & 0.3306 & 0.4103 & 0.4741 \\

\midrule

\multirow{2}{*}{dpROF} & \multirow{2}{*}{Urban100}
& SSIM $\uparrow$
& \textbf{0.9426} & \textbf{0.9165} & \textbf{0.8791} & \textbf{0.8442} & \textbf{0.8181} & \textbf{0.7563} & \textbf{0.6838} & \textbf{0.6148} \\
& & LPIPS $\downarrow$
& \textbf{0.0261} & \textbf{0.0416} & \textbf{0.0689} & \textbf{0.1049} & \textbf{0.1296} & \textbf{0.1947} & \textbf{0.2750} & \textbf{0.3518} \\

\midrule

\multirow{2}{*}{ROF} & \multirow{2}{*}{Urban100}
& SSIM $\uparrow$
& 0.9406 & 0.9126 & 0.8748 & 0.8357 & 0.8064 & 0.7387 & 0.6611 & 0.5879 \\
& & LPIPS $\downarrow$
& 0.0290 & 0.0473 & 0.0762 & 0.1102 & 0.1377 & 0.2084 & 0.2887 & 0.3650 \\

\midrule

\multirow{2}{*}{dpROF} & \multirow{2}{*}{EuroSAT}
& SSIM $\uparrow$
& \textbf{0.9313} & \textbf{0.9022} & \textbf{0.8555} & \textbf{0.8407}
& \textbf{0.8145} & \textbf{0.7737} & \textbf{0.7222} & \textbf{0.6837} \\
& & LPIPS $\downarrow$
& \textbf{0.0438} & \textbf{0.0565} & \textbf{0.0803} & \textbf{0.1134}
& \textbf{0.1305} & \textbf{0.1783} & \textbf{0.2282} & \textbf{0.2693} \\

\midrule

\multirow{2}{*}{ROF} & \multirow{2}{*}{EuroSAT}
& SSIM $\uparrow$
& 0.9275 & 0.8960 & 0.8529 & 0.8083
& 0.7759 & 0.7643 & {0.7155} & 0.6768 \\
& & LPIPS $\downarrow$
& 0.0438 & 0.0623 & 0.0911 & 0.1247
& 0.1511 & 0.1880 & 0.2349 & 0.2749 \\

\midrule

\multirow{2}{*}{dpROF} & \multirow{2}{*}{Set14}
& SSIM $\uparrow$
& \textbf{0.9220} & \textbf{0.8912} & \textbf{0.8485} & \textbf{0.8139}
& \textbf{0.7871} & \textbf{0.7281} & \textbf{0.6644} & \textbf{0.6102} \\
& & LPIPS $\downarrow$
& \textbf{0.0610} & \textbf{0.0825} & \textbf{0.1155} & 0.1627
& 0.1890 & \textbf{0.2545} & 0.3269 & 0.3891 \\

\midrule

\multirow{2}{*}{ROF} & \multirow{2}{*}{Set14}
& SSIM $\uparrow$
& 0.9206 & 0.8892 & 0.8484 & 0.8079
& 0.7785 & 0.7214 & 0.6576 & 0.6023 \\
& & LPIPS $\downarrow$
& 0.0637 & 0.0880 & 0.1211 & \textbf{0.1575}
& \textbf{0.1859} & 0.2577 & \textbf{0.3253} & \textbf{0.3852} \\

\midrule

\multirow{2}{*}{dpROF} & \multirow{2}{*}{Set14 (4x)}
& SSIM $\uparrow$
& \textbf{0.9448} & \textbf{0.9196} & 0.8800 & \textbf{0.8784}
& \textbf{0.8570} & \textbf{0.8281} & \textbf{0.7870} & \textbf{0.7575} \\
& & LPIPS $\downarrow$
& 0.0645 & 0.1019 & 0.1715 & \textbf{0.1726}
& \textbf{0.2143} & 0.2553 & 0.3222 & 0.3607 \\

\midrule

\multirow{2}{*}{ROF} & \multirow{2}{*}{Set14 (4x)}
& SSIM $\uparrow$
& 0.9433 & 0.9175 & \textbf{0.8822} & 0.8456
& 0.8185 & 0.8238 & 0.7845 & 0.7544 \\
& &LPIPS $\downarrow$
& \textbf{0.0639} & \textbf{0.1004} & \textbf{0.1531} & 0.2074
& 0.2458 & \textbf{0.2480} & \textbf{0.3053} & \textbf{0.3435} \\

\bottomrule
\end{tabular}%
}
\end{table*}

\section{Cross-dataset evaluation}

We next test whether the model with parameters optimized for BSDS500 train set retains its performance on different datasets. We use datasets Kodak24, Urban100, EuroSAT, and Set14 and compare with the behavior of classical ROF as a baseline. Kodak24 is a standard benchmark dataset for image restoration and denoising, consisting of 24 high-quality natural images of size $768\times512$ pixels, containing a variety of natural scenes, textures, smooth regions, and sharp edges. On the other hand, Urban100 \cite{Huang2015} is an image restoration benchmark with different characteristics, including urban scenes with repetitive patterns, fine structures, sharp edges, and highly textured regions. EuroSAT \cite{EuroSATpaper} is a dataset consisting of 27000 $64 \times 64$ patches of satellite photos of terrain of various types (forests, rivers, industrial etc.) -- images with vastly different characteristics from BSDS500 and other tested datasets. Set14 is a classical dataset containing natural images. We sourced it from \cite{Set14github}, where it is available along with a 4x upscaled version. We use it for to check the behavior of our model in higher resolution. 

We compare the performance in terms of SSIM, PSNR, and LPIPS metrics. We report the results in Table \ref{tab:veccross}. (We disclude PSNR here, as it gives analgous results to SSIM. Full results available in the 'RGB' supplement in \cite{github}.) We observe consistent supremacy of dpROF over ROF on Kodak24, Urban100 and EuroSAT. On Set14, the dpROF outperforms ROF in terms of SSIM, but the LPIPS results are not consistent. In any case, similar picture is observed on upscaled images, indicating stability of dpROF across image resolution. 

\section{Conclusions}

The adaptive double-phase ROF model that we introduced was capable of significantly reducing the staircasing effect compared to the classical ROF as well as the Huber-ROF variant. At the same time, it maintained comparable or better values of similarity metrics, and computational cost comparable to the Huber-ROF model.

In quantitative tests, dpROF consistently outperformed ROF and Huber-ROF, as well as TGV, Chen--Levine--Rao and non-local means in terms of SSIM and PSNR, while maintaining comparable values of LPIPS. The relative edge of dpROF was most pronounced for high levels of noise.  The performance of the model trained on BSDS500 remains stable among various datasets with different characteristics. 

Our results also suggest that using the ROF minimizer to construct the regularizer instead of the mollified input might improve the performance of other adaptive models. Moreover, it seems likely that using the Huber-TV integrand in place of the TV integrand as the basis of the adaptive regularizer could yield better results. We leave testing these ideas to possible future work. 

\noindent {\bf \flushleft Acknowledgments.}
The authors would like to express their sincere gratitude to the anonymous reviewers for their helpful comments and suggestions.

The first author would like to thank Ilaria Perugia and Monica Nonino for helpful discussions. 

The test images used in Section \ref{sec:qual} are: 
\begin{itemize} 
\item 'carygrant': a frame with the actor Cary Grant from the film \emph{His Girl Friday} (1940); in public domain in the US since 2024 \cite{wikicary}; 
\item 'girlface': a classical test image, presumed public domain; 
\item 'schnitzel': a photograph of a Wiener schnitzel with lemon, lamb's lettuce and potato salad, and cranberries; taken by M{\L}, public domain. 
\end{itemize}
512x512 px versions of the images were used. 

The code was implemented in Python 3.0. The quantitative comparison in Sections \ref{sec:comparison} and \ref{sec:multi} was carried out locally on a computer with the following parameters: AMD Ryzen 9 9900X processor (4.4GHz) and 64GB DDR5 RAM (6000MHz, CL30). 

\noindent {\bf \flushleft Funding.} The work of the first author has been supported by the Austrian Science Fund (FWF), grants 10.55776/ESP88 and 10.55776/I5149. The second author has been partially supported by the grant 2024/55/D/ST1/03055 of the National Science Centre (NCN), Poland. The third author has been partially supported by the Special Account for Research Funding of the National Technical University of Athens. For the purpose of open access, the authors have applied a CC BY public copyright license to any Author Accepted Manuscript version arising from this submission.

\bibliographystyle{asdfgh}
\bibliography{bib.bib}

\end{document}